\documentclass[letterpaper,twocolumn,10pt]{article}
\usepackage{usenix-2020-09}

\usepackage{tikz}
\usepackage{amsmath}
\usepackage{color}
\usepackage{array}
\usepackage{colortbl}
\usepackage{filecontents}
\usepackage{paralist}

\usepackage{algorithm}
\usepackage{algorithmic}
\usepackage{dsfont}
\usepackage{amsmath}
\usepackage{amsthm}
\usepackage{multirow}
\usepackage{caption}
\usepackage{subfigure}
\usepackage{xspace}
\usepackage{amsfonts}
\usepackage{booktabs}

\DeclareMathOperator*{\argmax}{argmax}

\newcommand{\tabincell}[2]{\begin{tabular}{@{}#1@{}}#2\end{tabular}}
\newcommand{\set}[1]{\{#1\}}
\newtheorem{definition}{Definition}

\newtheorem{theorem}{Theorem}

\newtheorem{remark}[theorem]{Remark}
\newcommand{\archname}{{Split-AI}\xspace}
\newcommand{\sysname}{{SELENA}\xspace}

\begin{document}
\date{}

\title{\Large \bf Mitigating Membership Inference Attacks by Self-Distillation\\  Through a Novel  Ensemble Architecture}

 \author{
 {\rm Xinyu Tang}\\
 Princeton University
 \and
 {\rm Saeed Mahloujifar}\\
 Princeton University
  \and
  {\rm Liwei Song}\\
 Princeton University
 \and
 {\rm Virat Shejwalkar}\\
 University of Massachusetts Amherst
 \and
 {\rm Milad Nasr}\\
 University of Massachusetts Amherst
 \and
 {\rm Amir Houmansadr}\\
 University of Massachusetts Amherst
 \and
 {\rm Prateek Mittal}\\
 Princeton University
 }

\maketitle

\begin{abstract}
Membership inference attacks are a key measure to evaluate privacy leakage in machine learning (ML) models. These attacks aim to distinguish training members from non-members by exploiting differential behavior of the models on member and non-member inputs.
The goal of this work is to train ML models that have high  membership privacy while  largely preserving their utility; we therefore aim for an empirical membership privacy guarantee as opposed to the  provable privacy guarantees provided by techniques like differential privacy, as such techniques are shown to deteriorate model utility. Specifically, we propose a new framework to train privacy-preserving models that induces similar behavior on member and non-member inputs to mitigate  membership inference attacks.
Our framework, called \sysname, has two major components. 
The first component and the core of our defense is a novel ensemble architecture for training. This architecture, which we call \archname,  
splits the training data into random subsets, and trains a model on each subset of the data. We use an adaptive inference strategy at test time: our ensemble architecture aggregates the outputs of only those models that did not contain the input sample in their training data. We prove that our \archname architecture defends against a large family of membership inference attacks, however, it is susceptible to new adaptive attacks. 
Therefore, we use a second component in our framework called Self-Distillation 
to protect against such stronger attacks. 
The Self-Distillation component (self-)distills the training dataset through our \archname ensemble, without using any external public datasets. 
Through extensive experiments  on major benchmark datasets we show that \sysname  presents a superior trade-off between membership privacy and utility compared to the state of the art. 
In particular, \sysname incurs no more than 3.9\% drop in classification accuracy compared to the undefended model. Compared to two state-of-the-art empirical defenses to membership privacy, MemGuard and adversarial regularization, \sysname reduces the membership inference attack advantage over a random guess by a factor of up to 3.7 compared to MemGuard and a factor of up to 2.1 compared to adversarial regularization.

\end{abstract}


\section{Introduction}
\label{sec:intro}

Machine learning has achieved tremendous success in many areas, but it requires access to data that may be sensitive.  Recent work has shown that machine learning models are prone to memorizing sensitive information of training data incurring serious privacy risks~\cite{shokri2017membership, carlini2019secret, carlini2020extracting, fredrikson2015model, salem2020updates,song2017machine,ganju2018property}. Even if the model provider is trusted and only provides query services via an API, i.e., black-box model access in which only prediction vectors are available, private information can still be obtained by attackers. The \emph{membership inference attack (MIA)} is one such threat in which an adversary tries to identify whether a target sample was used to train the target machine learning model or not based on model behavior\cite{shokri2017membership}.
MIAs pose a severe 
privacy threat by revealing private information about the training data. For example, knowing the victim's presence in the hospital health analytic training set reveals that the victim was once a patient in the hospital. 

Shokri et al.~\cite{shokri2017membership} conducted MIAs against machine learning in the black-box manner. They formalize the attack as a binary classification task and utilize a neural network~(NN) model along with shadow training technique to distinguish members of training set from non-members. Following this work, many MIAs have been proposed which can be divided into two categories: direct attacks~\cite{yeom2020overfitting, song2020systematic, yeom2018privacy, song2019privacy,nasr2018machine,nasr2019comprehensive}, which directly query the target sample and typically utilize only a single query; indirect attacks~\cite{long2018understanding, li2020label, choo2020label}, which query for samples that are in the neighborhood of the  target sample to infer membership, and typically utilize multiple queries. The research community has further extended the MIA to federated settings~\cite{melis2019exploiting, nasr2019comprehensive} and generative models~\cite{hayes2019logan}. MIAs have also provided a foundation for more advanced data extraction attacks~\cite{carlini2020extracting} 
and for benchmarking privacy-preserving mechanisms~\cite{jayaraman2019evaluating,nasr2021adversary}.

The effectiveness of MIAs and the resulting privacy threat has  motivated the research community to design several defense mechanisms against these attacks~\cite{abadi2016deep, nasr2018machine, jia2019memguard, shejwalkar2019reconciling}. As MIAs distinguish members and non-members of the target model based on the difference in model's behavior on members, defense mechanisms need to enforce similar model behavior on members and non-members. 
There exist two main categories of membership inference defenses, as shown in  Table~\ref{tab:prov_empi}: techniques that offer \emph{provable privacy}, and defenses that offer \emph{empirical membership privacy}. The first category mainly uses differential privacy mechanisms~\cite{abadi2016deep,mcmahan2017learning, wang2019subsampled} to be able to provide a \emph{provable privacy guarantee} for all inputs. However, the use of DP (e.g., in DP-SGD~\cite{abadi2016deep}) is shown to significantly reduce the utility of the underlying models in many machine learning tasks~(see Section~\ref{subsec:dpsgd}).
This has motivated the second category of membership inference defenses, where privacy is empirically evaluated through practical MIAs with the aim of preserving model utility. 
Our work in this paper falls in the second category, and as we will show, our technique offers a superior trade-off between MIA protection and model utility compared to the state-of-the-art empirical privacy defenses~\cite{nasr2018machine,jia2019memguard,shejwalkar2019reconciling} 
 (see Section~\ref{sec: eval} for more details).

\begin{table}[ht]
   \aboverulesep=0ex 
   \belowrulesep=0ex 
    \centering
    \begin{tabular}{c|c|c}
    \toprule
         &Low utility &High utility \\
         \midrule
         \rule{0pt}{1.1EM}
         \tabincell{c}{Provable\\ privacy}& \tabincell{c}{ DP-based:\\DP-SGD~\cite{abadi2016deep}} &\tabincell{c}{Desired~(No method\\ achieves this goal so far)}\\
         \midrule
        \tabincell{c}{Empirical \\membership \\privacy} &\tabincell{c}{Not be \\considered} &\tabincell{c}{ Adversarial \\Regularization~\cite{nasr2018machine},\\MemGuard~\cite{jia2019memguard},\\SELENA(Our work)}\\
         \bottomrule
         
    \end{tabular}
    \caption{Two categories of membership inference defenses: provable privacy with low utility vs. empirical membership privacy with high utility.}
    \label{tab:prov_empi}
\end{table}

\noindent\textbf{Our Framework.}
In this paper, we introduce a novel empirical  MIA defense framework, called \sysname,\footnote{SELf ENsemble Architecture.} whose goal is to protect against practical black-box MIAs while also achieving high classification accuracy. Our framework consists of two core components: \emph{\archname} and \emph{Self-Distillation}.

\textbf{Split Adaptive Inference Ensemble (\archname):}
Our first component, called \archname, is proposed to enable the model to have similar behavior on members and non-members. We obtain this goal by training multiple models~(called sub-models) with random subsets from the training set. 
While such ensemble architectures have been considered in different ML contexts, our framework's novelty lies in the particular adaptive approach it uses to respond to the queries. 
The key intuition is that for a training sample, if one of sub-models is not trained with it, this sub-model will have similar behavior on this training sample and other non-members. We use this intuition in our adaptive inference strategy. When the queried sample is in the training set, the adaptive inference procedure will only call the sub-models that did not use the query in their training set. When the queried sample is not in the training set, we query a particular subset of sub-models (as explained later in Section \ref{sec: ourdefense}).
Our approach provides an intuitive foundation for membership privacy: no matter if the queried sample is a member or a non-member, the adaptive inference will always use only those sub-models which have not used that sample for their training; this ensures membership privacy which we demonstrate through a formal analysis.

\textbf{Self-Distillation:} 
Our \archname shows promising performance against the traditional type of MIA, i.e., the direct single-query attack~\cite{song2019privacy, song2020systematic,yeom2020overfitting, yeom2018privacy}. However, it falls short in protecting against recent adaptive MIAs, which work by crafting multiple, particularly-fabricated queries~\cite{li2020label, choo2020label}. Moreover, \archname has a high computational overhead, as it needs to search for each queried sample within the training set,  and perform inference on multiple sub-models. 
To protect \archname against adaptive attacks and to reduce its computational overhead, we use the second component of our framework, which we call Self-Distillation.  Our Self-Distillation component  performs a novel form of knowledge transfer on the model created by \archname to produce a final \emph{protected model}. Specifically, it first  queries \archname with its exact training samples to get their corresponding prediction vectors. 
Then, it uses these prediction vectors as the soft labels of the training set to train the protected model, which will be used for inference. 
During the inference stage, the protected model only need to perform a single computation for each queried sample, therefore it has a much lower overhead compared to \archname's model.
Furthermore, the protected model protects not only against traditional single-query MIA attacks, but also against adaptive MIA attacks as shown in our analysis in Section~\ref{sec: eval}.  Note that, unlike conventional uses of distillation for membership privacy~\cite{shejwalkar2019reconciling}, our Self-Distillation component does not need a public dataset for knowledge transfer as it uses its own training dataset for (self-)distillation.\footnote{Note that our usage of the term self distillation is different from what Zhang et al.\cite{zhang2019your} refer to as self-distillation.}

\noindent\textbf{Evaluation.} We evaluate our \sysname on three benchmark datasets (CIFAR100, Purchase100, Texas100) and compare with existing defenses~\cite{nasr2018machine, jia2019memguard, song2020systematic} in a rigorous manner using two types of existing attacks and one type of adaptive attack. (1) We first analyze our defense by \emph{direct single-query attacks}, which have been typical used in most previous MI attacks and defenses. (2) We next evaluate our framework  by \emph{label-only attacks}, which infer membership information only based on labels and hence simply obfuscating prediction confidence vector can not protect against such attacks. (3) We finally study \emph{adaptive attacks}, which are tailored to our defense mechanism. Overall, \sysname achieves a better trade-off between the utility, i.e., classification accuracy, and the practical membership privacy without requiring additional public data. For utility,  \sysname incurs only a little drop in classification accuracy compared to the  undefended model~(no more than 3.9\%), and outperforms adversarial regularization~\cite{nasr2018machine} by up to 7.0\%~(on Texas100). For membership privacy risks, \sysname reduces the MIA advantage over a random guess by a factor of up to 4.0 compared to undefended model, a factor of up to 3.7 compared to MemGuard~\cite{jia2019memguard} and a factor of up to 2.1 compared to adversarial regularization~\cite{nasr2018machine}. 
Unlike DP-SGD that offers a provable privacy guarantee, our approach only provides an empirical membership  inference defense (similar to MemGuard and adversarial regularization). However, our evaluation shows that \sysname achieves a much better utility than DP-SGD~(See Figure~\ref{fig:compare}).
\begin{figure}
    \centering
    \includegraphics[width=3.0in]{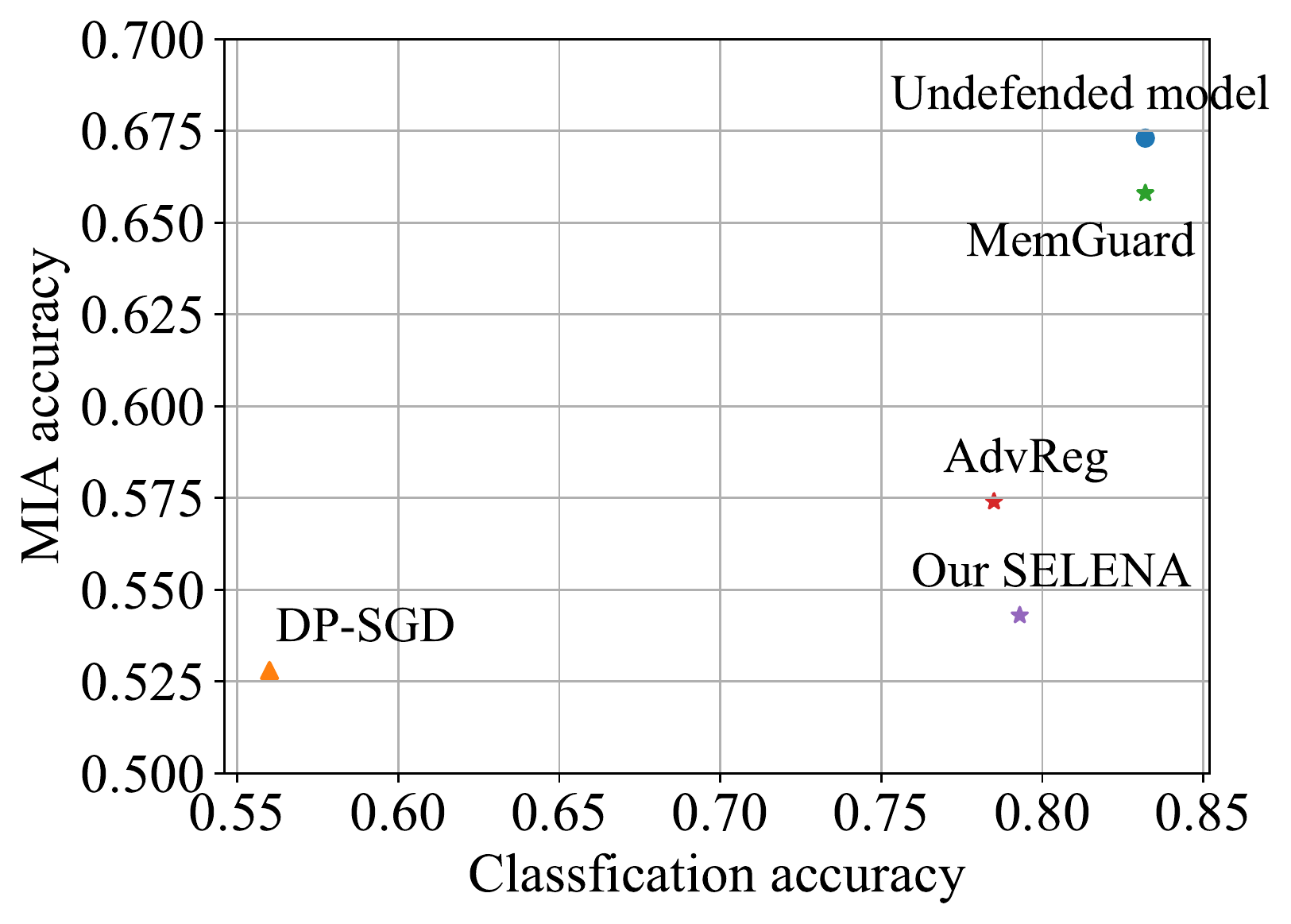}
    \caption{Comparison of our method with undefended model, DP-SGD~\cite{abadi2016deep}~($\epsilon=4$), MemGuard~\cite{jia2019memguard} and adversarial regularization~\cite{nasr2018machine} with respect to classification accuracy and MIA accuracy on Purchase100 dataset. Our SELENA outperforms adversarial regularization in both classification and MIA accuracy. Our SELENA significantly reduces MIA accuracy compared to undefended model and MemGuard while incurring little classification accuracy drop. Our SELENA achieves much higher classification accuracy compared to DP-SGD while only incurs little additional practical membership risks.}
    \label{fig:compare}
\end{figure}

In summary, we propose a membership inference defense to achieve high classification accuracy and highly mitigate practical MIAs. Our key contributions are as follows:
\begin{compactitem}
    \item We propose  \archname as the first component of our framework that enforces the model to have a similar behavior on members and non-members while maintaining a good classification accuracy using sub-models trained on overlapping subsets of data and an adaptive inference strategy. We further prove that the direct single-query attack can not achieve higher attack accuracy than a random guess against this component.
    \item We introduce Self-Distillation of the training set as the second component of our framework to overcome the limitations of the \archname while largely preserving its defense abilities without relying on an additional public dataset.
    \item We systematically evaluate our framework on three benchmark datasets including Purchase100, Texas100, and CIFAR100 against a suite of MIAs including direct single-query attacks, label-only (indirect multi-query) attacks and adaptive attacks to show that our framework outperforms prior defenses.

\end{compactitem}

\section{Preliminaries and Problem Formulation}
\label{sec:prelim} 

In this section, we introduce the machine learning concepts and notation relevant to our work, as well as our threat model and design goals.

\subsection{ML Preliminaries and Notation}
\label{pre}
    In this paper, we consider supervised machine learning for classification. Let  $F_{\theta}:\mathbb{R}^d \mapsto \mathbb{R}^k$ be a classification model with $d$ input features and $k$ classes, which is parameterized by $\theta$. For a given example $\textbf{z} = (\textbf{x}, y)$, $F_{\theta}(\textbf{x})$ is the classifier's confidence vector for $k$ classes and the predicted label is the corresponding class which has the largest confidence score, i.e., $\hat{y} =  \mathop{\argmax}_i F_{\theta}(\textbf{x})$.
    
    The goal of supervised machine learning is to learn the relationship between training data and labels and generalize this ability to unseen data. The model learns this relationship by minimizing the predicted loss across the training set $D_{tr}$:
    $$\min_{\theta}\frac{1}{|D_{tr}|}\sum_{\textbf{z}\in D_{tr}}l(F_{\theta}, \textbf{z})$$
    Here $|D_{tr}|$ is the size of the training set and $l(F_{\theta}, \textbf{z})$ is the loss function. 
    When clear from the context, we use $F$, instead of $F_\theta$, to denote the target model.

\subsection{Threat Model}
\textbf{Black-box attack:} In this paper, we follow previous defenses~\cite{nasr2018machine, jia2019memguard} and assume the attacker has black-box access to the target model, i.e., the attacker can only make queries to the model provider and obtain corresponding prediction vectors or predicted labels, 
instead of having access to target model's parameters. Therefore, the adversary can perform standard black-box attacks, in particular the \emph{direct single-query} attacks, which directly query the target sample \emph{one time} and is the typical benchmarking technique, and the \emph{label-only attacks}, which \emph{make multiple queries} for a single target sample and exploit predicted label information. 
{We also introduce a third type of black-box attack which is an adaptive attack tailored to our system. See Section~\ref{sec:attackdefense} for a detailed explanation of direct single-query attacks and label-only attacks and Section~\ref{sec:adaptiveattack} for the adaptive attacks.}

\textbf{Partial knowledge of membership for members:} 
Like previous defenses~\cite{nasr2018machine,song2020systematic}, we assume the adversary knows a small ratio of samples from the training set, i.e., it knows some members. The goal of the adversary is to identify any other member sample.

\subsection{Design Goals}
\label{subsec:goal}

In this paper, we aim to overcome the limitations of existing membership inference defenses~\cite{nasr2018machine, jia2019memguard, shejwalkar2019reconciling}, which estimate the membership risk through practical MIAs: none of these defenses are able to provide sufficient MIA protection and high utility simultaneously in the absence of public datasets.

\textbf{Low MIA accuracy:} Our goal is to design practical defenses against MIAs. We will evaluate our defense in a systematic and rigorous manner to ensure that it achieves low MIA accuracy (i.e., high membership privacy) across a broad class of attacks, instead of only one specific family of attacks.

\textbf{High classification accuracy:} We aim to protect membership privacy without significantly decreasing the classification accuracy (model utility).

\textbf{No additional public data required for defense:} Some prior works~\cite{papernot2016semi, shejwalkar2019reconciling} have proposed to preserve membership privacy by knowledge distillation using  publicly available datasets. However, this is a limiting assumption since public datasets may not be available in many real-world ML training scenarios such as healthcare data. In this paper, we consider a more realistic scenario, where the model provider does not have access to external public dataset.

\section{Existing Attacks and Defenses}
\label{sec:attackdefense}
Next, we overview prior MI attacks and MI defenses. 

\subsection{Membership Inference Attacks (MIAs)}
\label{subsec: attack}
    MIAs can utilize the prediction vector as a feature using a neural-network-based model, called \emph{NN-based attacks}, or can compute a range of custom metrics (such as correctness, confidence, entropy) over the prediction vector to infer membership, called \emph{metric-based attacks}. These attacks can be mounted either by knowing a subset of the training set~\cite{nasr2018machine} or by knowing a dataset from the same distribution of the training set and constructing shadow models~\cite{shokri2017membership}.

    Let us denote $D_{tr}$ as the training set for the target model, i.e., members and $D_{te}$ as the test set, i.e., non-members. $D_{tr}^A$ and $D_{te}^A$ are, respectively, the sets of members and non-members that the attacker knows. $I(\textbf{x}, y, F(\textbf{x}))$ is the binary membership inference classifier in the range of $\{0, 1\}$ which codes members as 1, and non-members as 0. The literature typically measures MIA efficacy as the attack accuracy:
     $$\frac{\sum_{(\textbf{x},y)\in D_{tr}\backslash D_{tr}^A}I(\textbf{x}, y , F(\textbf{x}))+\sum_{(\textbf{x},y)\in D_{te}\backslash D_{te}^A}(1-I(\textbf{x}, y , F(\textbf{x})))}{|D_{tr}\backslash D_{tr}^A| + |D_{te}\backslash D_{te}^A|}$$
    
    In most previous attacks~\cite{shokri2017membership, nasr2018machine, yeom2020overfitting, song2020systematic}, the number of members and non-members used to train and evaluate the attack model are the same. With this approach, the prior probability of a sample being either a member or a  non-member is 50\% (corresponding to a random guess).
    
    Next, we summarize  black-box MIAs in the following two categories: \textbf{direct} attacks and \textbf{indirect} attacks.
    
    \textbf{Direct single-query attacks:} Most existing MIAs directly query the target sample and utilize the resulting prediction vector. Since ML models typically have only one output for each queried sample, just a single query is sufficient.

    \emph{ NN-based attack~\cite{shokri2017membership, nasr2018machine}:} 
    The attacker can use the prediction vectors from the target model along with the one-hot encoded ground truth labels as inputs and build a NN model~\cite{nasr2018machine} $I_{\text{NN}}$ for the membership inference task.  
    
    \emph{Correctness-based attack~\cite{yeom2020overfitting}:} Generalization gap~(i.e., the difference between training accuracy and test accuracy) is a simple baseline for MIA as samples with correct prediction are more likely to be training members. 
    $$I_{\text{corr}}(F(\textbf{x}),y) = \mathds{1}\{\mathop{\argmax}_i F(\textbf{x})_i =y\}$$

    \emph{Confidence-based attack~\cite{yeom2018privacy, song2019privacy, song2020systematic}:} Prediction confidence corresponding to training samples $F(\textbf{x})_{{y}}$ is typically higher than prediction confidence for testing samples. Therefore, confidence-based attack will only regard the queried sample as a member when the prediction confidence is larger than either a class-dependent threshold $\tau_y$ or a class-independent threshold $\tau$.
    $$I_{\text{conf}}(F(\textbf{x}),y)= \mathds{1}\{ F(\textbf{x})_{y} \geq {\tau_{(y)}}\}$$

    \emph{Entropy-based attack~\cite{shokri2017membership, song2020systematic}}: The prediction entropy of a training sample is typically lower than the prediction entropy of a testing sample.
    Therefore, entropy-based attack will only regard the queried sample as a member when the prediction entropy is lower than a class-dependent threshold $\tau_y$ or a class-independent threshold $\tau$.

    $$I_{\text{entr}}(F(\textbf{x}),y) = \mathds{1}\{-\sum_i F(\textbf{x})_i \log(F(\textbf{x})_i) \leq \tau_{(y)}\}$$    
    
    \emph{Modified entropy-based attack~\cite{song2020systematic}:}
    Song et al.~\cite{song2020systematic} proposed the modified prediction entropy metric which combines the information in the entropy metric and ground truth labels.:
    \begin{equation*}
    \begin{split}
        \text{Mentr}(F(\textbf{x}),y) &=  -(1-F(\textbf{x})_y)\log (F(\textbf{x})_y)\\ &- \sum_{i\neq y} F(\textbf{x})_i \log(1-F(\textbf{x})_i)
    \end{split}
    \end{equation*}
    Training samples typically have lower values of modified entropy metric than testing samples and either a class-dependent threshold $\tau_y$ or a class-independent threshold $\tau$ attack is applied to infer membership:
    $$I_{\text{Mentr}}(F(\textbf{x}),y) = \mathds{1}\{\text{Mentr}(F(\textbf{x}),y) \leq \tau_{(y)}\}$$

    \textbf{Indirect multi-query attacks~(label-only attacks):}
    Long et al.~\cite{long2018understanding} stated that indirect attacks can make queries that are related to target sample \textbf{x} to extract additional membership information as a training sample influences the model prediction both on itself and other samples in its neighborhood. These indirect attacks usually make multiple queries for a single target sample~\cite{long2018understanding,li2020label, choo2020label}. For example, multi-query \emph{label-only attacks} leverage the predicted label of the queried data as features, and are thus immune to defenses that only obfuscate prediction confidences, e.g., MemGuard~\cite{jia2019memguard}. The key idea in label-only attacks is that the model should be more likely to correctly classify the samples around the training data than the samples around test data, i.e., members are more likely to exhibit high robustness than non-members~\cite{li2020label, choo2020label}. Simply obfuscating a model’s confidence scores can not hide label information to defend against such label-only attacks. 

    \emph{Boundary estimation attacks~\cite{li2020label,choo2020label}:} As the target model is more likely to correctly classify the samples around training samples than those around test samples, the distance to classification boundary for training sample should be larger than that for the test samples. An attacker can either leverage techniques for finding adversarial examples under the black-box assumption~\cite{chen2020hopskipjumpattack,brendel2017decision} or add noise to find the adversarial examples that change the predicted label with minimum perturbation. Such attacks should not achieve higher attack accuracy than the white-box adversarial examples attack such as Carlini-Wagner attack~\cite{carlini2017towards}, which has full access to the model parameters and can find the adversarial example with the least distance for each target sample.

    \emph{Data augmentation attacks~\cite{choo2020label}:}
    In computer vision tasks, data augmentation techniques based on translation, rotation, and flipping help improve test accuracy. However, such data augmentation techniques pose a  privacy threat: the target model is more likely to correctly classify the augmented data of training samples. An attacker can query the augmented data of the target record and use the percentage of correct predictions to identify membership of the target record.

\subsection{Existing Defenses}
\label{subsec:defense}
    Multiple defenses have been proposed to mitigate MIAs. We summarize them below. Section~\ref{sec:related} gives a more comprehensive summary of prior defenses.
    
    \textbf{Adversarial Regularization~\cite{nasr2018machine}:} Nasr et al.~\cite{nasr2018machine} include the estimation of membership threat in the training process of the ML model. They optimize a min-max game to train a privacy-preserving target classifier, which aims to reduce the prediction loss while also minimizing the MIA accuracy.

    \textbf{Early Stopping ~\cite{caruana2001overfitting,song2020systematic}:} During the training process, the model may learn too much information in the training samples thus the difference between its behavior on members and non-members becomes larger and larger, and the model becomes more vulnerable to MIAs. Therefore, early stopping, which is a general technique to prevent model overfitting by stopping model training before the whole training process ends, can mitigate MIA accuracy with a sacrifice of model utility. Song et al.~\cite{song2020systematic} find that adversarial regularization is not better than early stopping~\cite{caruana2001overfitting} when evaluated by a suite of attacks including both NN-based attacks and metric-based attacks. They recommend that any defense that trades off a reduction in MIA accuracy at the cost of a reduction in utility should be compared with early stopping as a baseline. 
    
    \textbf{MemGuard~\cite{jia2019memguard}:} Jia et al.~\cite{jia2019memguard} obfuscate the prediction vector with a well-designed noise vector using the  perspective of adversarial examples to confuse the membership inference classifier. Since MemGuard doesn't change prediction results, and only obfuscates confidence information, it maintains the original classification accuracy of the undefended model.  
    Song et al.~\cite{song2020systematic} shows that MemGuard~\cite{jia2019memguard} lacks consideration of strategic adversaries; though resistant to NN-based attack, MemGuard underestimates the threat of metric-based attacks.
    
    \textbf{DP-based defenses:} Differential privacy~\cite{dwork2008differential} is a formal framework that provides a rigorous privacy guarantee. In machine learning, DP-based defenses add noise to the training process of a classifier such as DP-SGD~\cite{abadi2016deep}. However, it is challenging to perform machine learning with differential privacy while achieving acceptable utility loss and privacy guarantees\cite{jayaraman2019evaluating, rahman2018membership}~(See Section~\ref{subsec:dpsgd} and Section~\ref{sec:related}).
\section{Our Defense Architecture}
\label{sec: ourdefense}
\begin{figure*}[ht]
\centering
\includegraphics[width = 6 in]{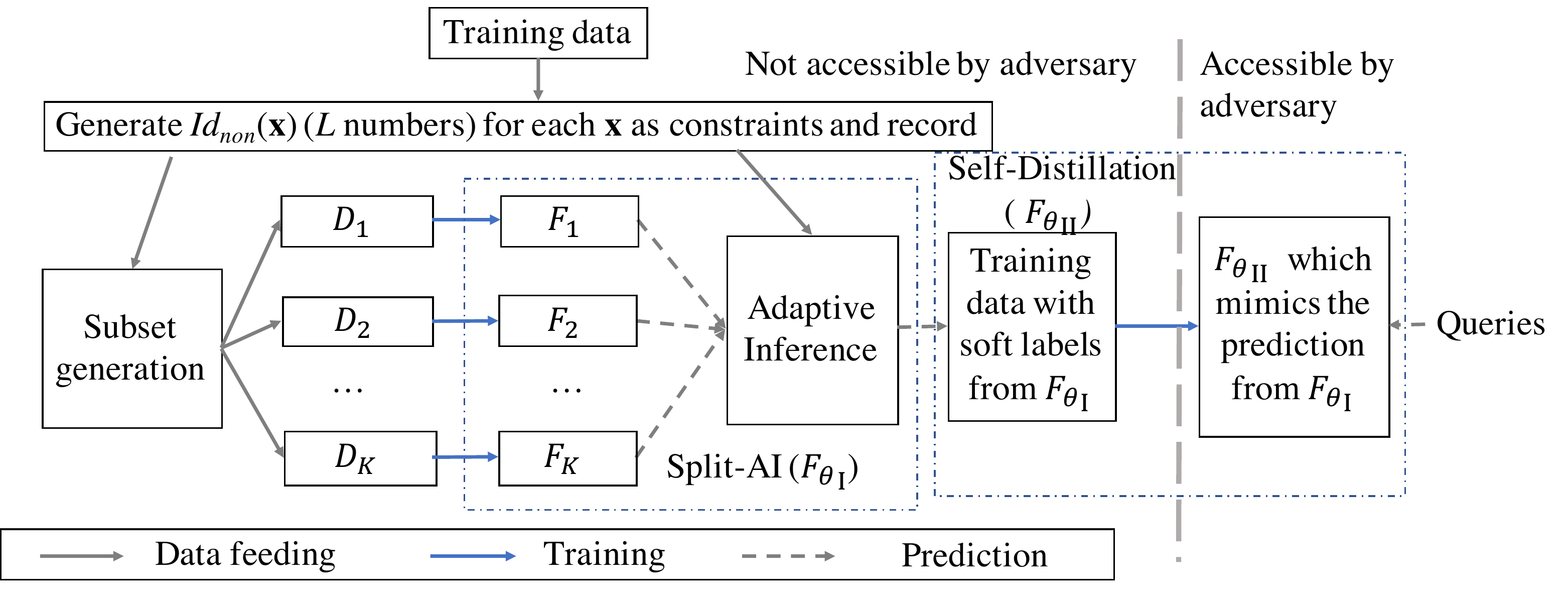}
\caption{Our end-to-end defense framework with the \archname and Self-Distillation components.}
\label{fig:system}
\end{figure*}
\label{sec:methodology}

In this section, we first present an overview of our defense framework and then describe the two key framework components: \archname and Self-Distillation. 

\subsection{Overview}
\label{subsec:overview}

MIAs aim to distinguish members and non-members of the private training data of a model. These attacks use the fact that the trained model has a different behavior on member and non-member data. This difference in behavior can appear in different forms, for example, the accuracy of model might be different on members and non-members~\cite{salem2018ml}, or the confidence might be higher on member inputs~\cite{yeom2018privacy,song2019privacy, song2020systematic}. Similarly, the model might be more likely to correctly classify the samples around the member examples compared to those around non-member examples~\cite{choo2020label, li2020label}. MIAs leverage these differences to obtain an attack advantage that is better than a random guess even in the black-box setting. Current defenses typically consider adding specific constraints during the optimization process, either in the training phase\cite{nasr2018machine} or in the inference phase\cite{jia2019memguard}, to reduce the mismatch of model behavior on members and non-members. However, optimization under multiple constraints for machine learning, which is usually non-convex, is a computationally hard task. We instead propose a framework to defend against MIAs by training multiple sub-models using subsets from whole training set and introducing a specific adaptive inference technique that exploits the following intuition: \emph{if a training sample is not used to train a sub-model, that sub-model will have similar behavior on that training sample and non-members.}~Section~\ref{subsec:results} shows the advantage of our defense in improving the trade-off between membership privacy and utility, which is based on this intuition, over existing membership inference defenses~(MemGuard~\cite{jia2019memguard} in Section~\ref{subsubsec:memguard} and adversarial regularization~\cite{nasr2018machine} in Section~\ref{subsubsec:advreg}).

Based on this intuition, we propose a framework, \sysname, composed of two components to defend against MIAs. The first component, which we call \emph{\archname}, trains an ensemble of $K$ sub-models with overlapping subsets of the training dataset. The constraint for each subset is as follows: for a given training sample, there are $L$ sub-models which are not trained with that training sample, and therefore, they will behave similarly on that training sample and the non-member samples. \archname applies adaptive inference for members and non-members: For a member sample, \archname  computes $L$ predictions of the $L$ sub-models which are not trained with the member sample, and  outputs the average of the $L$ predictions as the final prediction. For a non-member sample, the adaptive inference  randomly samples $L$ sub-models from the $K$ total sub-models, subject to a specific distribution, and returns the average of the $L$ predictions on the non-member as the final prediction. We detail our algorithm and explain why it preserves membership privacy in Section~\ref{subsec:se}.

The second component, which we call \emph{Self-Distillation},  addresses the two weaknesses of  \archname: its potential privacy risks due to multi-query/adaptive attacks and its high inference-time computation overhead. Specifically, the Self-Distillation component transfers the knowledge of the model obtained by \archname into a new model by using the soft labels of the training set from \archname. We call this Self-Distillation because it does not require any public dataset for distillation. As we will demonstrate, this protected model from Self-Distillation has similar classification performances as \archname with significantly lower computation overhead,  and can protect against advanced multi-query and adaptive MIA attacks. 

As we study the black-box MIAs, only the final prediction vectors or predicted labels of the protected model from Self-Distillation are available to the attacker. Figure~\ref{fig:system} gives an overview of our defense, where we denote \archname as $F_{\theta_{\mathrm{I}}}$ and protected model from Self-Distillation as $F_{\theta_{\mathrm{II}}}$.\footnote{PATE\cite{papernot2016semi, papernot2018scalable} also trains multiple sub-models to provide privacy but with a public dataset. We detailed the difference between our SELENA and PATE in Section~\ref{subsec:PATE}.} Next, we detail \archname and Self-Distillation.

\subsection{Our \archname Ensemble Architecture}
\label{subsec:se}
Here we describe \archname, the first component of our system.

\paragraph{\archname's training:}
Following the intuition in Section~\ref{subsec:overview}: we train $K$ sub-models and ensure that each training sample is not used to train $L$ sub-models such that these $L$ sub-models will have similar behavior on this training sample and other non-members. We accomplish this via a specific data partitioning strategy:

\emph{For each data point $\textbf{x}$ in the training set, we randomly generate $L$ non-model indices from \{$1$, 2, ..., $K$\} to denote the $L$ non-models that are not trained with the data point and record the identification numbers of these $L$ non-model indices~(denoted as ${Id}_{{non}}(\textbf{x})$\footnote{${Id}_{{non}}(\textbf{x})$ records $L$ sub-model indices which are not trained with $\textbf{x}$.}). We then generate the dataset partition based on these non-model indices. For each subset $D_i$, we will only use those training samples which do not include $i$ in their non-model indices.}

    \begin{figure}[htbp]
        \centering
        \subfigure[\archname's training]{
        \begin{minipage}[t]{0.5\linewidth}
        \centering
        \includegraphics[width=1.5in]{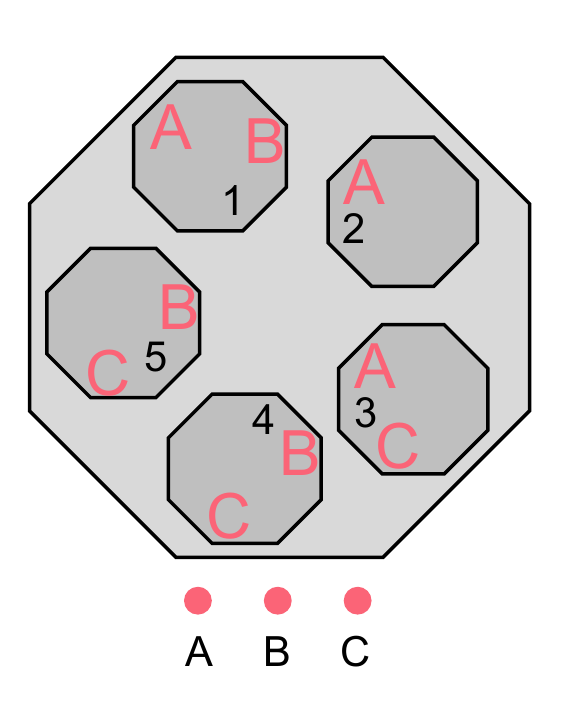}
        \label{fig:3a}
        \end{minipage}
        }%
        \subfigure[\archname's inference]{
        \begin{minipage}[t]{0.5\linewidth}
        \centering
        \includegraphics[width=1.5in]{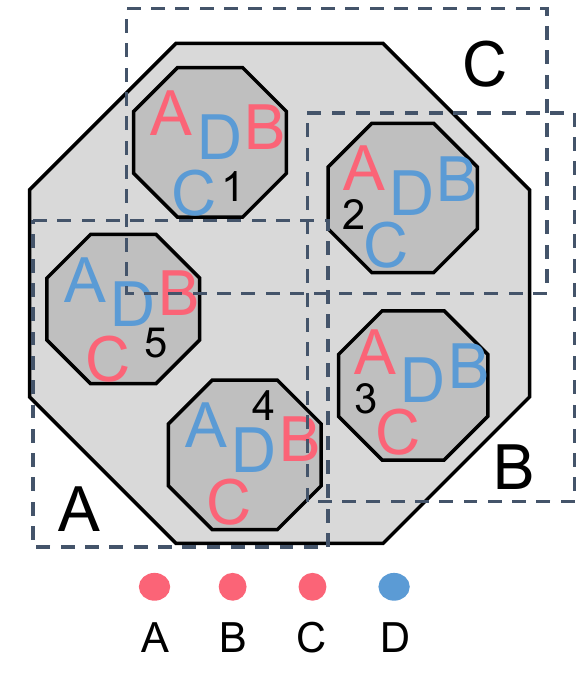}
        \label{fig:3b}
        \end{minipage}
        }%

    \caption{Illustration of \archname's data partition for training and adaptive inference in $K=5, L=2$ for member samples $A,B,C$ and non-member sample $D$~(red color for members and blue color for non-members). ${Id}_{{non}}(A)=(4,5), {Id}_{{non}}(B)=(2,3), {Id}_{{non}}(C)=(1,2)$. }
    \label{fig:splitai}
    \end{figure}

Figure~\ref{fig:3a} illustrates this partition strategy for three training samples~($A$, $B$, $C$) under the setting of $K=5,L=2$.
We randomly generate non-member sub-model indices: ${Id}_{{non}}(A)=(4,5)$, ${Id}_{{non}}(B)=(2,3)$, ${Id}_{{non}}(C)=(1,2)$. Therefore, $A$ is used to train sub-model 1, 2, 3. $B$ is used to train sub-model 1, 4, 5. $C$ is used to train sub-model 3, 4, 5.

This specific data partition strategy ensures that for each data point, we have $L$ sub-models which are not trained with it. This facilitates our key intuition in Split-AI: we use models that are not trained with a data point to estimate its soft label while protecting the membership information. $K$ and $L$ are parameters of our framework. The approximate size of each subset is $((K-L)/K) \times |D_{tr}|$. We then train $K$ sub-models $F_{i}$, one for each subset of the training data $D_i$, which have the same architecture and hyper-parameter settings.

\paragraph{\archname's inference:} We now describe the adaptive inference based ensemble strategy for members and non-members. For each queried sample $\textbf{x}$, the ensemble will check whether there is an exact match of $\textbf{x}$ in the training set:
\begin{compactitem}
    \item If so, which indicates that $\textbf{x}$ is a member, the defender will average the prediction vectors on $\textbf{x}$ from $L$ models which are not trained with $\textbf{x}$ as the output;
    \item If not, the defender will randomly use non-member indices of a member sample $\textbf{x}'$ and average the prediction vectors on $\textbf{x}$ from $L$ models of ${Id}_{non}(\textbf{x}')$ as the output.
\end{compactitem}

Figure~\ref{fig:3b} illustrates the adaptive inference for three member samples~$(A$, $B$, $C)$ and one non-member sample $D$ following the setting in Figure~\ref{fig:3a}. $A$ is non-member for sub-model 4, 5; $B$ is non-member for sub-model 2, 3; $C$ is non-member for sub-model 1, 2; and $D$ is non-member for all sub-models. Adaptive inference will average on non-member indices in sub-models for $A$, $B$, $C$ and randomly select one member sample's non-member indices for non-member sample $D$.

\begin{algorithm}[t]
\caption{\archname Model ${F_{\theta_{\mathrm{I}}}}$}
\label{alg:A1}
\begin{algorithmic}
\STATE {Initialize:\\ $K$: total number of sub-models $F_{1}, F_{2}, ..., F_{K}$\\
$L$: for each training sample, the number of sub-models which are not trained with it.
\\$(X_{train}, Y_{train})$: training data and labels}
\STATE{\textbf{Training Phase:} }
\STATE{Randomly generate the $L$ non-model indices for each training sample ${Id}_{non}(\textbf{x})$.}
\FOR {$i = 1$ to $K$}
\STATE{Construct subset $(X_{train}^i, Y_{train}^i)$ for model $F_{i}$ based on the recorded ${Id}_{non}$s for models: \{$(\textbf{x},y)$: $(\textbf{x},y) \in (X_{train}, Y_{train})$, $i$ not in ${Id}_{non}(\textbf{x})$\}}

\FOR{\text{number of the training epochs} }
\STATE{Update $F_{i}$ by descending its stochastic gradients over  $l(F_{i}(X_{train}^i ),Y_{train}^i ) $.}
\ENDFOR
\ENDFOR
\STATE{\textbf{Inference Phase:} $F_{\theta_\mathrm{I}}(\textbf{x}$)}
\STATE{Given $\textbf{x}$}
\IF{$\textbf{x}$ in $X_{train}$}
\STATE{$$F_{\theta_{\mathrm{I}}}(\textbf{x}) = \frac{1}{L}{\sum_{i\in {Id}_{non}(\textbf{x})}}F_{{i}}(\textbf{x})$$}
\ELSE
\STATE{Randomly select $\textbf{x}'$ in the training set, 
$$F_{\theta_{\mathrm{I}}}(\textbf{x}) = \frac{1}{L}{\sum_{i\in {Id}_{non}(\textbf{x}')}}F_{{i}}(\textbf{x})$$}

\ENDIF
\end{algorithmic}
\end{algorithm}

Algorithm~\ref{alg:A1} presents the entire procedure for \archname. We formally prove that \archname strategy is resilient to \emph{direct} single-query MIAs~(discussed in Section~\ref{subsec: attack}) and can reduce the accuracy of such attacks to a random guess by Theorem~\ref{thm:single_query_split_train} in Appendix~\ref{appendix:splitaidsq}, which provides the theoretical foundation of our defense capability. The intuitive explanation for this proof is that for each data point, the distribution of output of this algorithm on this given point $\textbf{x}$ is independent of the presence of $\textbf{x}$ in the training set. This is because, we will not use models that are trained with $\textbf{x}$ to answer queries, even if $\textbf{x}$ is in the training set. Our evaluation on \archname in Appendix~\ref{appendix:splitai} is consistent with Theorem~\ref{thm:single_query_split_train}: \archname  maintains a good classification accuracy by leveraging flexible overlapping subsets and an ensemble of $L$ sub-models.

\subsection{Our Self-Distillation Mechanism}
\label{subsec:sd}

\paragraph{Limitations of \archname.} While \archname is resilient to direct single-query MIAs, an adversary can leverage more advanced attacks. For example, instead of direct query, attacker may do an indirect query~\cite{long2018understanding} for the target sample (see Section~\ref{subsec: attack} for definitions) or may do multiple queries for one target sample to identify membership information, as suggested in recent work~\cite{choo2020label}. \archname suffers from severe privacy risks under the setting of such aggressive attacks that exploit the matching process for training samples: (1) An adaptive attacker can make a single \emph{indirect} query by adding a small noise to the target sample. Such adaptive attacks can fool the matching process used in the inference strategy that checks if the input is a training member or non-member. \archname will recognize noisy training samples as non-members and may end up using sub-models trained with the target sample, thus leaking membership information. (2) Attacker can perform replay attacks by making multiple queries for the same target sample: \archname will only have one possible prediction vector for members, while approximately $C_K^L$ possible prediction vectors for non-members. Furthermore, \archname incurs a computational overhead during inference:
For each queried sample, \archname first needs to identify whether it is in the training set, thus  incurring overhead for this matching process. Second, \archname needs to perform inference on $L$ models for each queried sample, while conventional approaches only perform inference on a single model. 

\paragraph{Self-Distillation.} To overcome the above limitations, we need a more sophisticated defense mechanism and we correspondingly introduce the second component of our framework. We leverage distillation, which is proposed by Hinton et al.~\cite{hinton2015distilling} to reduce the size of NN architectures or ensemble of NN architectures. To be more specific, here we use a method which we call \emph{Self-Distillation}: we first apply \archname  to get the prediction vectors for the training samples. We then use the same training set along with the prediction vectors (obtained from \archname) as soft labels to train a new model using conventional training. The new protected model benefits from distillation to
to maintain a good classification accuracy. For queried samples, the defender now just need to do the inference on the new protected model $F_{\theta_{\mathrm{II}}}$ distilled from the \archname.  For defense capability, we prove that this new model largely preserve \archname's defense ability against direct single-query attack by Theorem~\ref{thm:distill} and Corollary~\ref{cor:last} under mild stability assumptions~(Definition~\ref{def:stabledistillation}) in Appendix~\ref{appendix:selenadsq}. Note that our theoretical analysis of \sysname is only valid for single-query direct attacks. In fact, there exist some datasets that \sysname cannot obtain provable privacy for under multi-query attacks. This includes settings with similar data points that have different labels~(see Appendix~\ref{appendix:correlationdiscuss}).

\textbf{Self-Distillation overcomes the privacy limitations of \archname and mitigates advanced MIAs.} 
The defender controls the Self-Distillation component and ensures that Self-Distillation only queries each exact training sample once. The attacker only has black-box access to the protected output model of Self-Distillation, but cannot access the Split-AI model. Hence, the attacker cannot exploit the soft labels computation of Split-AI as discussed before. Hence, the final protected model from Self-Distillation effectively mitigates the replay and multi-query indirect attacks:

    1. For replay attack: each sample is only queried once during the Self-Distillation process, while replay attack requires at least two queries of each sample to obtain advantage over random guess. In addition, the final protected model has a deterministic behavior with only one possible prediction vector for each queried sample.
    
    2. For single-query indirect attacks: each exact sample is queried during the Self-Distillation process and noisy samples around the training sample are not queried. In addition, the attacker only has black-box access to the protected model from Self-Distillation~(and no access to defender's \archname): indirect query attacks are thus limited in obtaining additional membership information~(See Appendix~\ref{appendix_subsec:necessityofdistill} for more details).

Self-Distillation also solves the computational overhead of the \archname in inference: the defender now does not need to check whether the queried sample is a training sample and it only needs to make inference on a single Self-Distilled model.

In Section~\ref{sec: eval}, we will evaluate the effectiveness of our whole framework via rigorous experimental analysis including direct single-query attacks, label-only attacks and adaptive attacks.

\section{Membership Inference Attacks Evaluated}
\label{sec:adaptiveattack}

To (empirically) demonstrate the privacy of our system, we evaluate it  against three main classes of MIA attacks. 
First, we evaluate our attack against the single-query and label-only MIA attacks introduced in earlier.
Specifically, evaluate against the direct single-query attacks in Section~\ref{subsec: attack}, and  for label-only attacks, we  use the boundary attack for all three datasets and the data augmentation attack for CIFAR100. We describe these attacks in a detailed manner in Appendix~\ref{appendix:label}. 

Additionally, we evaluate our system against \emph{adaptive membership inference attacks}, as introduced in the following. 
Song and Mittal~\cite{song2020systematic} emphasizes the importance of placing the attacker in the last step of the arms race between attacks and defenses: the defender should consider adaptive attackers with knowledge of the defense to rigorously evaluate the performance of the defenses. Therefore, here we consider attacks that are tailored to our defense. As our defense leverages soft labels from the \archname ensemble to train a new model $F_{\theta_{\mathrm{II}}}$ in Self-Distillation, we need to analyze whether and how an attacker can also leverage the information about soft labels.

    We first note that an attacker is unable to directly interact with our \archname to directly estimate soft labels, since the prediction API executes queries on the model produced by the Self-Distillation component. Second, we expect that when the model provider finishes training the protected model $F_{\theta_{\mathrm{II}}}$ with soft labels obtained from \archname, it can safely delete the sub-models and soft labels of the training set to avoid inadvertently leaking information about the soft labels. 
    However, an attacker can still aim to indirectly \emph{estimate} soft labels.

    As we assume that the attacker knows partial membership of the exact training set in evaluating membership privacy risks~(specifically, half of the whole training set) and attacker cannot have access to the defender's non-member model indices ${Id}_{non}({\textbf{x}})$ for training set, the attacker will generate new non-member model indices ${Id}_{non}({\textbf{x}})'$ for these known member samples to train a new shadow \archname ensemble and use the shadow \archname to estimate soft labels of the target samples. 
    The attacker can then use such soft labels as an additional feature to learn the difference in target model's behavior on members and non-members, and launch MIAs on $F_{\theta_{\mathrm{II}}}$.
    The shadow \archname discussed in our paper is stronger than original shadow models~\cite{shokri2017membership} since it is trained with exact knowledge of the partial training dataset.
    
    We design four adaptive direct single-query attacks\footnote{Our Table~\ref{tab:allattacks} shows that label-only attacks is weaker than direct single-query attacks on undefended model. We have also designed adaptive multi-query label-only attacks against \sysname and evaluated on Purchase100 dataset, which is better than original label-only attacks, but weaker than adaptive direct single-query attacks.} including two NN-based attacks and two metric-based attacks to leverage the information estimated soft labels. To clarify, $F_{\theta_{\mathrm{II}}}$ denotes the protected target model which answers the attacker's queries and $F_{\theta_{\mathrm{I}}}'$ denotes attacker's shadow \archname.
    
    \textbf{MIAs based on NN and soft labels}: 
    The first NN-based attack concatenates the soft labels obtained from $F_{\theta_{\mathrm{I}}}'$, the predicted confidence from $F_{\theta_{\mathrm{II}}}$ and the one-hot encoded class labels as features to train a neural network attack model~(denoted as $I_{\text{NN1}}$). The second attack utilizes the {difference between} the estimated soft labels from $F_{\theta_{\mathrm{I}}}'$ and outputs from $F_{\theta_{\mathrm{II}}}$, and uses this difference as an input to the neural network architecture used by Nasr et al.~\cite{nasr2018machine}~(denoted as $I_{\text{NN2}})$.
    
    \textbf{MIAs based on distance between soft labels and predicted confidence}: Similar to previous metric-based attacks~\cite{song2020systematic}, an attacker may try to distinguish between members and non-members by leveraging the distance between estimated soft labels from $F_{\theta_{\mathrm{I}}}'$, and the prediction confidence vectors from $F_{\theta_{\mathrm{II}}}$. We have:
    $$I_{\text{dist}}(F_{\theta_{\mathrm{II} }}(\textbf{x}), F_{\theta_\mathrm{I}}'(\textbf{x}),y ) = \mathds{1}\{ Dist(F_{\theta_{\mathrm{II}}}(\textbf{x}), F_{\theta_{\mathrm{I}}}'(\textbf{x})) \leq  \tau_{(y)}\}$$
    $$\text{or, } I_{\text{dist}}(F_{\theta_{\mathrm{II} }}(\textbf{x}), F_{\theta_\mathrm{I}}'(\textbf{x}),y ) = \mathds{1}\{ Dist(F_{\theta_{\mathrm{II}}}(\textbf{x}), F_{\theta_{\mathrm{I}}}'(\textbf{x})) \geq  \tau_{(y)}\}$$
    where we apply both class-dependent threshold $\tau_y$ and class-independent threshold $\tau$ and we will report the highest MIA accuracy. In this work we consider $L_2$ distance $I_{\text{L}_\text{2}\text{-}dist}$ and cross-entropy loss $I_{\text{CE-}dist}$~(since the cross-entropy loss function is used for training our defense models).

\section{Evaluations} 
\label{sec: eval}
    
    In this section, we first briefly introduce the datasets and model architectures used to train the classification  models in Section~\ref{sec: eval-setup}. More details can be found in Appendix~\ref{appendix:setup}.

    Next we systematically evaluate our end-to-end defense framework including its efficacy against (1) direct single-query attacks, (2) indirect label-only attacks, and (3) adaptive attacks and make a comparison with undefended model, MemGuard~\cite{jia2019memguard}, adversarial regularization~\cite{nasr2018machine} and early stopping~\cite{song2020systematic} by considering both the utility and membership privacy risks in Section~\ref{subsec:results}.
    \subsection{Experimental Setup} 
    \label{sec: eval-setup}

    We use three benchmark datasets and target models which are widely used in prior works on MI attacks and defenses.
    
    \textbf{Datasets.} Purchase100, Texas100 and CIFAR100. We follow Nasr et al.~\cite{nasr2018machine} to determine the partition between training data and test data and to determine the subset of the training and test data that constitutes attacker's prior knowledge. Specifically, the attacker's knowledge corresponds to half of the training and test data, and the MIA success is evaluated over the remaining half. 
    
    \textbf{Target Models.} For CIFAR100, we use ResNet-18~\cite{he2016deep}, which is a benchmark machine learning model widely used in computer vision tasks. For Purchase100 and Texas100, we follow previous work~\cite{nasr2018machine} to use a 4-layer fully connected neural network with layer sizes $[1024,512,$ $256,100]$.
    
    In our defense, we set $K=25$ and $L=10$ for all three datasets. To show that our defense is effective across multiple model architectures and settings of $K$ and $L$, we vary activation functions, width, depth of target models, as well as different choices of $K$ and $L$ and present the results in Appendix~\ref{appendix:ablation}. We will release code to reproduce all our experiments.

    \begin{table*}[ht]
        \caption{Comparison of membership privacy and accuracy on training/test set of undefended model, previous defenses and  \sysname on three different datasets. AdvReg refers to adversarial regularization. The last column is the highest attack accuracy for each row, i.e. for a specific defense on one dataset, the highest attack accuracy that MIAs can achieve. The last column gives an overview of comparison: the lower the best attack accuracy, lower the membership inference threat. For each dataset, the defense which has the lowest corresponding attack accuracy is bold in the column of best direct single-query attack, best label-only and best attack.}
        
        \label{tab:allattacks}
        \centering
        \begin{tabular}{cccccccc}
        \toprule
        dataset&defense &\tabincell{c}{acc on \\training set} &\tabincell{c}{acc on \\test set}&\tabincell{c}{best direct \\single-query\\ attack}&\tabincell{c}{best\\label-only \\attack}&\tabincell{c}{best adaptive \\attack} &\tabincell{c}{best attack}\\
        \midrule
        \multirow{4}{*}{Purchase100} &None &99.98\% &83.2\% &{{67.3\%}} &65.8\%  &N/A &67.3\%\\
        &MemGuard&99.98\% &83.2\% &58.7\% &{65.8\%} &N/A &65.8\%\\
        &AdvReg &91.9\% &78.5\% &57.3\% &{57.4\%} &N/A &57.4\%\\
        &\textbf{\sysname} &82.7\% &79.3\% &\textbf{53.3\%} &\textbf{53.2\%} & {54.3\%} &\textbf{54.3\%}\\
        \midrule
        \multirow{4}{*}{Texas100} &None &79.3\% &52.3\% &{66.0\%} &64.7\% &N/A &66.0\% \\
        &MemGuard &79.3\% &52.3\% &63.0\% &{64.7\%} &N/A &64.7\%\\
        &AdvReg &55.8\% &45.6\% &{60.5\%} &56.6\% &N/A &60.5\%\\
        &\textbf{\sysname} &58.8\% &52.6\% &\textbf{54.8\%} &{\textbf{55.1\%}} &{54.9\%} &\textbf{55.1\%}\\
        \midrule
        \multirow{5}{*}{CIFAR100} &None &99.98\% &77.0\% &{74.8\%} &69.9\% &N/A &74.8\%\\
        &MemGuard &99.98\% &77.0\% &68.7\% &{69.9\%} &N/A &69.9\%\\
        &AdvReg &86.9\% &71.5\% &58.6\% &{59.0\%} &N/A &59.0\%\\
        &\textbf{\sysname} &78.1\% &74.6\% &\textbf{55.1\%} &\textbf{54.0\%} & {58.3\%}  &\textbf{58.3\%}\\ 
        \bottomrule
        \end{tabular}
    \end{table*}

    \begin{figure*}[htbp]
        \centering
        \subfigure[Purchase100]{
        \begin{minipage}[t]{0.3\linewidth}
        \centering
        \includegraphics[width=2.1in]{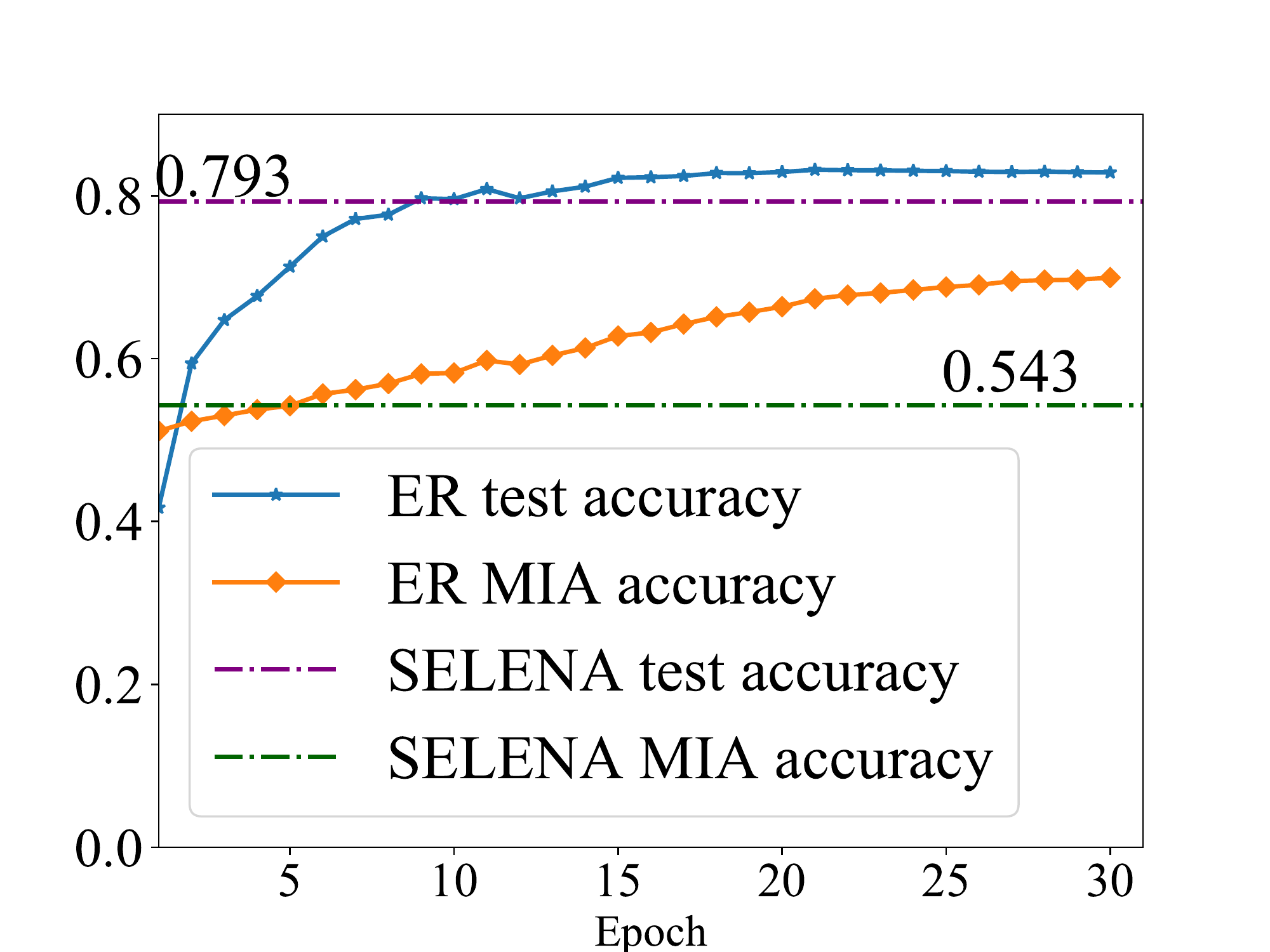}
        \end{minipage}
        }
        \subfigure[Texas100]{
        \begin{minipage}[t]{0.3\linewidth}
        \centering
        \includegraphics[width=2.1in]{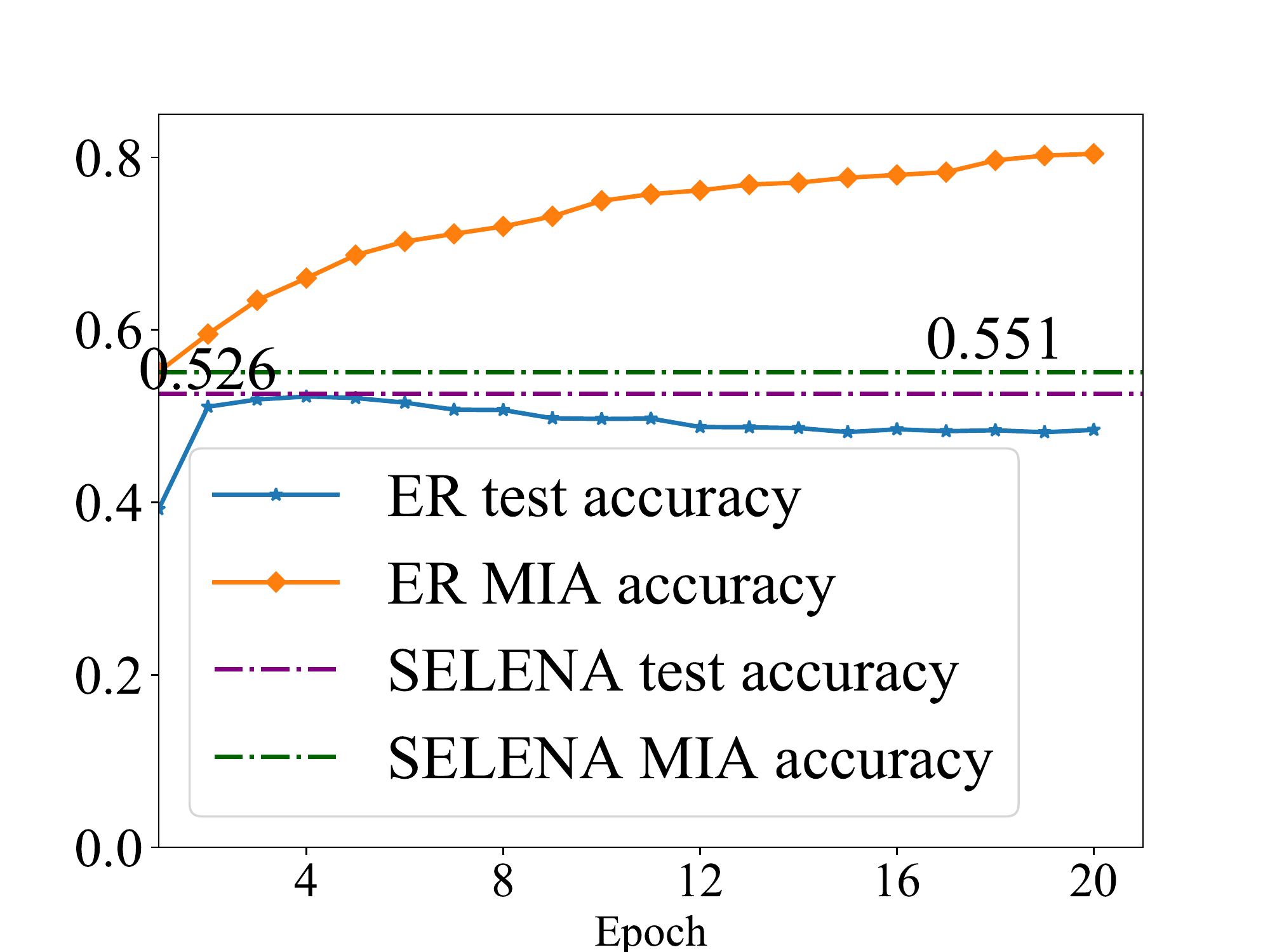}
        \end{minipage}
        }
        \subfigure[CIFAR100]{
        \begin{minipage}[t]{0.3\linewidth}
        \centering
        \includegraphics[width=2.1in]{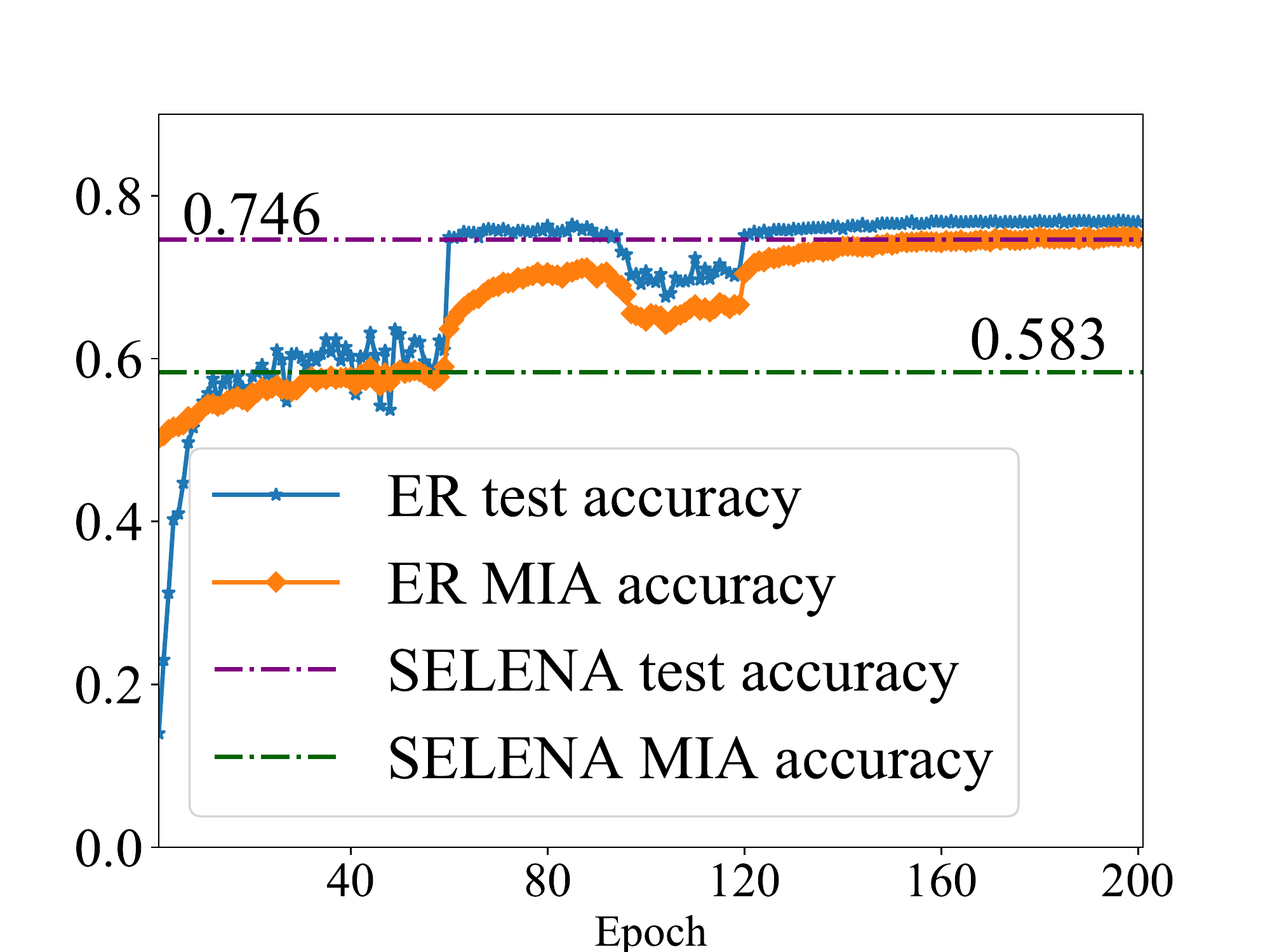}
        \end{minipage}
        }
        \centering
    \caption{Detailed comparison of \sysname with early stopping. From left to right are results for Purchase100, Texas100 and CIFAR100. The solid curves are the test accuracy and MIA accuracy with corresponding training epochs. ER denotes early stopping. The dashed lines are the test accuracy and MIA accuracy of \sysname, which is shown in Table~\ref{tab:allattacks}. Our defense achieves a better privacy-utility trade-off than all epochs in the conventional training.}
    \label{fig:earlystopping}
    \end{figure*}

    \subsection{Results}
    \label{subsec:results}
    
    Table~\ref{tab:allattacks} summarizes the classification accuracy and best attack accuracy for each attack type, including comparison with both undefended models~(in Section~\ref{subsubsec:undefended}) and previous defenses~(MemGuard in Section~\ref{subsubsec:memguard} and adversarial regularization in Section~\ref{subsubsec:advreg}). In addition, we also compare our SELENA with early stopping in Section~\ref{subsubsec:earlystop}. 

    \subsubsection{Comparison with Undefended Model}\label{subsubsec:undefended}
    
    We first compare our SELENA with undefended model on both membership privacy threats and classification accuracy.
    
    \textbf{\sysname significantly reduces membership inference risks.} From Table~\ref{tab:allattacks}, we can see that our defense leads to a significant reduction in privacy risks. Across three types of attacks, the MIA accuracy against our defense is no higher than 54.3\% on Purchase100, 55.1\% on Texas100 and 58.3\% on CIFAR100. 
    On the other hand, MIA accuracy against undefended models (in the absence of our defense) is much higher: such MIA advantage over a random guess is a factor of $3.0\sim 4.0$ higher than our defense.

    \textbf{\sysname achieves its privacy benefits at the cost of a small drop in utility.} Compared with undefended models, our defense only has a small utility loss (while providing substantial privacy benefits). Compared to undefended models, the classification(test) accuracy of our defense incurs at most $3.9\%$ accuracy drop~(on Purchase100), and even no accuracy drop on Texas100. We also discuss a more flexible trade-off between utility and membership privacy by combining the outputs from \archname and ground truth labels as the soft labels in Self-Distillation in Appendix~\ref{appendix:tradeoff}.

    We also note that even though our approach has a small loss in utility, it achieves a better utility-privacy trade-off compared to prior defenses like MemGuard, adversarial regularization and early stopping, which we discuss next.
    
    \subsubsection{Comparison with MemGuard}
    \label{subsubsec:memguard}
    \textbf{While the test accuracy of our defense is a little lower than MemGuard (MemGuard has the same test accuracy as the undefended model), the MIA accuracy against MemGuard is much higher than our defense.}  Compared to a random guess, which achieves 50\% attack accuracy, the best attacks on MemGuard can achieve $14.7\% \sim 19.9\%$ advantage over a random guess, 
    which is a factor of $2.4\sim 3.7$ higher than our defense. In general, MemGuard does not have any defense against MIAs that do not rely on confidence information: attacker can use label-only attacks as adaptive attacks since MemGuard only obfuscates confidence. 
    
    \subsubsection{Comparison with Adversarial Regularization}\label{subsubsec:advreg}
    \textbf{Our defense achieves higher classification accuracy and lower MIA accuracy compared with adversarial regularization.} 
    The classification accuracy of our defense is higher than adversarial regularization across all three datasets, and as high as 7.0\% for the Texas100 dataset. For MIAs, our defenses achieves significantly lower attack accuracy than adversarial regularization. MIA attacks against adversarial regularization is higher than our defense across all three datasets, and its advantage over random guess is at most a factor of 2.1 than our defense~(on Texas100). Besides, adversarial regularization is much harder to tune and can also take more training time (by a factor up to 7.8) compared to our defense when multiple GPUs are used in parallel~(see Section \ref{subsec: efficiency}).

    \subsubsection{Comparison with early stopping}\label{subsubsec:earlystop}
    We further compare our defense with early stopping, which can also help in minimizing the difference in model behavior on members and non-members~\cite{song2020systematic}. Specifically, we will compare the model performance of an undefended model in each epoch during the training process and our final protected model  $F_{\theta_{\mathrm{II}}}$. For early stopping, we only consider direct single-query attack~(due to their strong performance on undefended models). Figure~\ref{fig:earlystopping} shows a detailed comparison between our defense $F_{\theta_{\mathrm{II}}}$ and early stopping. The dashed lines are the classification accuracy on test set and the best MIA accuracy of our defense, which is already reported in Table~\ref{tab:allattacks}. 
    The solid lines correspond to classification accuracy on test set and MIA accuracy using the undefended model as a function of the training epochs. As we can see from Figure~\ref{fig:earlystopping}, \emph{our defense significantly outperforms early stopping.}
    
    \textbf{Comparison at similar attack accuracy.} The undefended model will only have same level of MIA accuracy as the dashed line of our defense at the very beginning of the training process. However the test accuracy of the undefended model at that point is far lower than that of our defense. For example, approximately, 
        \emph{for Texas100, when MIA accuracy against the conventional trained model is 55.1\%, the test accuracy of the undefended model is 13.4\% lower than that of our defense.} For other two dataset, when the MIA accuracy against the undefended model achieves similar attack accuracy as our defense, the test accuracy is 8.0\% lower on Purchase100 and 11.0\% lower on CIFAR100 compared to our defense.
       
    \textbf{Comparison at similar classification accuracy.} When the undefended model achieves the same classification accuracy on the test set as \sysname, the MIA accuracy against the undefended model is significantly higher than our defense. For example, \emph{when the test accuracy of the conventional model reaches 74.6\% on CIFAR100 (similar to our defense), the attack accuracy is 63.6\%, compared to the best attack accuracy of 58.3\% for our defense (which is 5.3\% lower).} We can see similar results on other datasets: when the test accuracy of undefended models achieves similar classification accuracy as our defense on Purchase100 and Texas100, the attack accuracy is 58.1\% on Purchase100 and 66.0\% on Texas100, which is 3.8\% and 10.9\% higher than \sysname respectively.

    We also highlight the following two points from Table~\ref{tab:allattacks}:

        1. Our \sysname effectively induces the similar behaviors including generalization, confidence, robustness for member and non-member samples and therefore the MIA attack accuracy is significantly reduced. Let us take the generalization gap $g$ as an example: in undefended models/MemGuard, $g$ is 16.78\% on Purchase100, 27.0\% on Texas100, 22.98\% on CIFAR100; in adversarial regularization, $g$ is 13.4\% on Purchase100, 10.2\% on Texas100 and 15.4\% on CIFAR100. In contrast, in our defense, $g$ is 3.4\% on Purchase100, 6.2\% on Texas100 and 3.5\% on CIFAR100: Our mechanism reduces the total generalization gap by a factor of up to 6.6 compared to undefended models/MemGuard, and a factor of up to 4.4 compared to adversarial regularization. 
        
        2. The additional estimation of soft labels provided by shadow \archname (using the entirety of the attacker's knowledge) provides additional information to the attacker which enhances the accuracy of our adaptive attacks: attack has more advantage over random guess than direct single-query attack and label-only attacks. However, even considering the strong adaptive attacks, \sysname still achieves lower attack accuracy in comparison to previous defenses, which validates the defense effectiveness of our \sysname. In addition, Appendix~\ref{appen:adaptiveattack} further analyze the membership privacy risks when the attacker knows different ratio of the training sets.

    In conclusion, using direct single-query attacks, label-only attacks, as well as adaptive attacks with estimated soft labels, we show that our approach outperforms previous defenses and achieves a better trade-off between utility and practical membership privacy. We also discuss that \archname can defend against direct single-query attack while maintaining a good classification accuracy, and the necessity for the adaptive ensemble in \archname and Self-Distillation in Appendix~\ref{appendix:components}.

\section{Discussions} 
    In this section, we will discuss the computation overhead of our defense, as well as comparison with PATE~\cite{papernot2016semi,papernot2018scalable}~(which uses a disjoint training set partition for sub-models and differential privacy to protect privacy), model stacking~\cite{salem2018ml}~(which uses a disjoint training set partition to train a hierarchical architecture, intuited by dropout~\cite{srivastava2014dropout}) and DP-SGD~\cite{abadi2016deep}, which provides a provable privacy guarantee for neural networks.
    
    \subsection{Efficiency}
    \label{subsec: efficiency}
    
    One cost that our framework needs to pay is the use of additional computing resources in the training process as we train multiple sub-models for \archname.     
    Table~\ref{tab:traintime} and Table~\ref{tab:inferencetime} present the comparison for the training time cost and inference time cost of our \sysname with previous defenses.
    As MemGuard~\cite{jia2019memguard} focuses on post-processing techniques for prediction vectors of undefended models in the inference phase, we omit MemGuard in Table~\ref{tab:traintime} for training time and compare our \sysname with MemGuard in Table~\ref{tab:inferencetime} for inference time. For time comparison, all experiments are tested on a single NVIDIA Tesla-P100 GPU. We separately set the batch size as 512, 128, 256 during training for Purchase100, Texas100 and CIFAR100~(note that the batch size might also impact the running time, and here we maintain the batch size for each dataset across different defenses the same). For undefended model, adversarial regularization, and our \archname, we train 30, 20, and 200 epochs for Purchase100, Texas100, and CIFAR100. For Self-Distillation, we train 60, 30, and 200 epochs for Purchase100, Texas100, and CIFAR100 to ensure convergence. All running times are tested three times and we report the average of test time.

    \begin{table}[ht]
        \caption{Comparison of training time.}
        \label{tab:traintime}
        \centering
        \begin{tabular}{ccccc}
             \toprule
             Dataset &None&\tabincell{c}{AdvReg}&\tabincell{c}{\sysname\\sequential} &\tabincell{c}{\sysname \\parallel}\\
             \midrule
             \tabincell{c}{Purchase\\100}&9.5s &55.7s  &359.4s &73.5s \\
             \midrule
             \tabincell{c}{Texas100}&10.7s &111.6s &343.0s &68.0s\\
             \midrule
             \tabincell{c}{CIFAR100} &1.78h  &23.5h  &29.6h &3.0h\\
            \bottomrule
        \end{tabular}
    \end{table}

    Comparison of training time in Table~\ref{tab:traintime}: our defense~(\sysname sequential: sequentially train each sub-model on a single GPU) costs up to 6.1h more compute time than adversarial regularization (CIFAR100). However, we can simply accelerate training \archname by training several sub-models parallelly. For example, if we train all $K$ sub-models simultaneously~(\sysname parallel), the training time for \sysname is 73.5s for Purchase100, 68.0s for Texas100 and 3.0h for CIFAR100. In contrast, adversarial regularization cannot benefit from parallel training: there is only one model during training.

    \begin{table}[ht]
        \caption{Comparison of inference time. Test on 1000 samples: 500 members and 500 non-members. Batch size is 1.}
        \label{tab:inferencetime}
        \centering
        \begin{tabular}{ccc}
             \toprule
            Dataset & MemGuard&\sysname  \\
             \midrule
             \tabincell{c}{Purchase100} &702.7s &0.7s\\
             \midrule
            \tabincell{c}{Texas100}  &668.6s &0.7s\\
            \midrule
            \tabincell{c}{CIFAR100}  &768.5s &8.6s\\
            \bottomrule
        \end{tabular}
    \end{table}

    Comparison of inference time in Table~\ref{tab:inferencetime}: 
    MemGuard cost three orders of magnitude more per inference compared to SELENA since it has to solve a complex optimization problem to obfuscate prediction vectors for every query while \sysname only needs to perform computation on a single model. 
    
    In conclusion, we argue that the cost of computing resources in the training phase and no additional computation in inference phase is acceptable as the improvement in GPU technology are making the computing resources cheap while the privacy threat remains severe. If multiple GPUs are available, our approach can easily benefit from parallelization by training the $K$ sub-models in parallel. Finally, we can also tune the system parameters $K$ and $L$ to control the trade-off between the training time resources, model utility and privacy.
    
    \subsection{Comparison with PATE} 
    \label{subsec:PATE}
    PATE~\cite{papernot2016semi,papernot2018scalable} is a framework composed of teacher-student distillation and leverages public data to achieve a better privacy-utility trade-off for differential privacy. PATE uses a disjoint training set partition for sub-models in the teacher component. 
    To get the private label of the public dataset to train the student model, PATE applies noisy count among sub-models.

    There are three major differences between our work and PATE: (1) PATE requires a \emph{public dataset} to provide the provable end-to-end privacy guarantee, which is not possible in certain practical scenarios such as healthcare. Our defense does not need public datasets and provides a strong empirical defense against MIAs. (2) We apply a novel \emph{adaptive inference strategy} to defend against MIAs: for each training sample, we only use prediction of sub-models in Split-AI that are not trained with it as these sub-models will not leak membership information for it. PATE does not use adaptive inference and relies on majority voting over all sub-models. (3) We use \emph{overlapping} subsets to train sub-models. This allows our approach to obtain high accuracy for each sub-model with sufficient subset size. PATE faces the limitation of each sub-model being trained with much reduced subset size due to disjoint subsets.

    In addition, PATE incurs a 0.7\% $\sim$ 6.7\% drop in test accuracy~\cite{papernot2018scalable}, while the test accuracy drop in our defense is no more than 3.9\%.

    \subsection{Comparison with Model Stacking}
    \label{subsec:modelstacking}
    Model stacking~\cite{salem2018ml} uses a two-layer architecture: the first layer contains a NN and a Random Forest,  and the combination of two outputs from the first layer is forwarded to a Logistic Regression model. Here we briefly discuss the difference between model stacking and our defense.
    
    Model stacking requires a disjoint subset of data for each model: NN and Random Forest in first layer and logistic regression in second layer to help improve membership privacy. This may cause a decrease in test accuracy of the overall ensemble. Also, directly combing the two outputs from the first layer as input to the second layer may still leak information even when Logistic Regression is trained on another subset of data: the membership inference risks of NN or Random Forest may be directly forwarded on to the Logistic Regression module. Appendix~\ref{appendix:modelstacking} presents the experiments to compare our defense and model stacking, which supports above two statements. We do not include model stacking in Section~\ref{sec: eval} as it does not achieve the state-of-the-art performance~\cite{jia2019memguard,nasr2018machine}.
    
    \subsection{Comparison with DP-SGD}
    \label{subsec:dpsgd}
    In this work, we use the canonical implementation of DP-SGD and its associated analysis from the TensorFlow Privacy library.\footnote{https://github.com/tensorflow/privacy.} We varied the parameter $noise\_multiplier$ in the range of [1, 3] on Purchase100 and [1, 2] on Texas100 with a step size 0.2. We set the privacy budget $\epsilon$ = 4 and report the best classification accuracy for these two datasets.

    The test accuracy on Purchase100 is 56.0\% and the best direct single-query MIA accuracy is 52.8\%. The test accuracy on Texas100 is 39.1\%, and the  best direct single-query MIA accuracy is 53.8\%. Note that though DP-SGD provides a  differential privacy guarantee and the best direct single-query MIA accuracy is  0.5\% $\sim$ 1\% lower than that against our \sysname, DP-SGD suffers from a significant loss in utility: compared to the undefended model DP-SGD incurs 13.2\% $\sim$ 27.5\% drop in classification accuracy, while our defense incurs no more than 3.9\% drop in test accuracy.

\section{Related Work}
\label{sec:related}
    \textbf{Membership inference attacks against machine learning.}
    MIAs are usually studied in a black-box manner~\cite{shokri2017membership, salem2018ml, nasr2018machine}: an attacker either leverages the shadow training technique or utilizes knowledge of partial membership information of training set.
    Most MIAs are direct single-query attacks~\cite{song2019privacy, song2020systematic,yeom2020overfitting, yeom2018privacy}. A more recent line of MIA research has considered indirect multi-query attacks which leverage multiple queries around the target sample to extract additional information~\cite{long2018understanding, choo2020label, li2020label, jayaraman2020revisiting}.
    Jayaraman et al.~\cite{jayaraman2020revisiting} analyze MIA in more realistic assumptions by relaxing proportion of training set size and testing set size in the MIA set up to be any positive value instead of $1$. Hui et al.~\cite{hui2021practical} study MIA in a practical scenario, assuming no true labels of target samples are known and utilizing differential comparison for MIAs.
    Another threat model for MIAs is that of a white-box setting, i.e., attacker has full access to the model~\cite{nasr2019comprehensive, leino2020stolen}, which can exploit model parameters to infer membership information.

   \noindent \textbf{Membership inference defenses for machine learning.} Membership inference defenses can be divided into two main categories. One category of defenses are specifically designed to mitigate such attacks. 
   It has been shown that techniques to improve a model's generalization ability, including regularization~\cite{krogh1992simple} and dropout~\cite{srivastava2014dropout}, can decrease the MIA success~\cite{shokri2017membership, salem2018ml} limitedly. Several defenses~\cite{nasr2018machine, li2020membership} propose to add a specific constraint during training to mitigate the difference of model behavior on models and non-models. These optimization problems under multiple constraints in training are usually computationally hard to solve. Post-processing techniques on prediction vectors are also applied on membership inference defenses~\cite{jia2019memguard, yang2020defending}. Note that these defenses which obfuscate prediction vectors can not defend against label-only attacks~\cite{li2020label, choo2020label}. Moreover, Song et al.~\cite{song2020systematic} re-evaluate two state-of-the-art defenses~(adversarial regularization~\cite{nasr2018machine} and MemGuard~\cite{jia2019memguard}) and find that both of them underestimated metric-based attacks. Shejwalkar et al.~\cite{shejwalkar2019reconciling} propose distillation of public data to protect membership privacy. However, public dataset is not usually available in many practical scenarios. Another category of defenses use differential privacy mechanisms~\cite{dwork2006calibrating,dwork2008differential, dwork2014algorithmic}, which provides a provable privacy guarantee for users. A general framework combining deep learning and differential privacy is DP-SGD~\cite{abadi2016deep, mcmahan2017learning, wang2019subsampled}. However, machine learning with differential privacy suffers from the challenge of achieving acceptable utility loss and privacy guarantee ~\cite{jayaraman2019evaluating, rahman2018membership}. Several methods have been proposed to improve test accuracy under an acceptable $\epsilon$ guarantee,
   which is still an active area of research. Current state-of-the-art approaches still incur significant drop in test accuracy~(around 25\%) on benchmark datasets with acceptable $\epsilon \leq 3$~\cite{tramer2020differentially, papernot2020tempered, nasr2020improving}.
   
    \noindent \textbf{Other Attacks Against Machine Learning Privacy.} Fredrikson et al.~\cite{fredrikson2015model} propose model inversion attacks, which can infer missing values of an input feature from the classifier's prediction. Ganju et al.~\cite{ganju2018property} study property inference attacks aiming to infer properties of the target model's training set. Salem et al.~\cite{salem2020updates} propose the dataset reconstruction attack in the online learning setting. {Another line of works studied model extraction attacks~\cite{tramer2016stealing, he2020stealing, krishna2019thieves}, i.e., stealing the ML model's learned parameters through the prediction API. Besides model parameters, other works focus on stealing the target model’s hyperparameters~\cite{wang2018stealing}.} Recently Carlini et al.~\cite{carlini2019secret, carlini2020extracting} studied memorization and data extraction attacks on natural language processing models, which show that machine learning models suffer from severe privacy threats.

\section{Conclusions}
In this paper we introduce a new practical membership inference defense using \archname and Self-Distillation. We first split the training set into $K$ subsets to train $K$ sub-models. We ensure each training sample is not used to train $L$ sub-models, and apply an adaptive inference strategy for members and non-members. \archname will only use the average of a particular subset of $L$ sub-models which are not trained with the queried samples. Hence \archname can defend against direct single-query attacks. We apply Self-Distillation from \archname to defend against stronger attacks and avoid additional computing resources in inference. We perform rigorous MIAs including direct single-query attacks, label-only attacks and adaptive attacks to show that our defense outperforms previous defenses to achieve a better trade-off between the utility and practical membership privacy. Future work includes understanding the adaptation of Split-AI in other privacy tasks such as provable private mechanisms, analyzing the defense performance against white-box MIAs, and extending the our defense from classification models to generative models.

\bibliographystyle{plain}
\bibliography{selena}

\appendix
\section{Proof for \archname and \sysname against Direct, Single-Query Membership Inference Attack}
\label{appendix:securityproof}
\newcommand{\Supp}{\mathrm{Supp}}
\paragraph{Notation.} In this section, we use $x\gets X$ to denote that $x$ is sampled from a distribution $X$. We use $\Supp(X)$ to denote the support set of a random variable $X$. By $TV(X,X')$ we denote the total variation distance between $X$ and $X'$, that is $TV(X,X') = \sup_{S\subset \Supp(X)\cup\Supp(X')} \Pr[X\in S] - \Pr[X'\in S]$.

\subsection{\archname's Privacy under Direct Single-query Attacks}
\label{appendix:splitaidsq}
\begin{definition}[Direct, Single-Query Membership Inference]\label{secGame}
The single-query membership inference game is defined between an attacker $A$ and a learner $C$ and is parameterized by a number $n$ which is the number of training examples. 
\begin{enumerate}
    \item The attacker selects a dataset $X=\set{x_1,\dots,x_{2n}}$ and sends it to the learner.
    \item Learner selects a uniformly random Boolean vector $b=b_1,\dots, b_{2n}$ such that the Hamming weight of $b$ is exactly $n$.
    \item Learner constructs a dataset $S=\set{x_i; \forall i\in [2n], b_i=1}$ and learns a model $F_{\theta_I}$ using $S$ as training set.
    \item Learner selects a random $i\in [2n]$ and sends $(x_i, F_{\theta_I}(x_i))$ to the adversary
    \item Adversary outputs a bit $b'_i$.
\label{def:securitygame}    
\end{enumerate}
The advantage of $A$ in breaking the security game above is $\mathsf{SQMI}(A,C,n)=\mathbf{E}[1-|b_i-b_i'|]$ where the expectation is taken over the randomness of the adversary and learner.
\end{definition}

\begin{remark}
We can  define a variant of the security game of Definition \ref{secGame} for a fixed dataset $X$. That is, instead of $X$ being chosen by adversary, we define the game for a given $X$. We use $\mathsf{SQMI}(A,C,X)$ to denote the success of adversary in the security game with the dataset fixed to $X$.
\end{remark}

\begin{theorem}\label{thm:single_query_split_train}
Consider a learner $C_{ST}$ that uses Algorithm \ref{alg:A1}. For any direct, single-query membership inference adversary $A$ we have

$$\mathsf{SQMI}(A,C_{ST},n) = 50\%$$
\end{theorem}

\begin{proof}[Proof] We show that for any adversary's choice of $i\in [2n]$ in step 4 of the security game, the view of adversary in two cases when $b_i=0$ and when $b_i=1$ are statistically identical. Note that the only information that the adversary receives is $r_i=F_{\theta_I}(x_i).$ We show that the distribution of two random variables $r_i \mid b_i=0$ and $r_i \mid b_i=1$ are identical. Let $U_i$ be a random variable corresponding to the subset of trained models that do not contain $x_i$ in their training set (in particular $|U_i|=L$ if $b_i=1$ and $|U_i|=K$ when $b_i=0$). Also, let $U$ denote a random variable corresponding to a subset of $L$ models that do not contain a random $x_k$ in their training data where $k$ is selected from $\set{j\in[2n]; b_j=1}$ uniformly at random.

We first note that $U\mid b_i=0$ and $U_i\mid b_i=1$ are identically distributed random variables. Specifically, they are both an ensemble of $L$ models trained on a uniformly random subset of a dataset $T\subset\set{x_1,\dots,x_{i-1},x_{i+1},\dots,x_{2n}}$ where $|T|=n-1$.

Now, lets calculate the distribution of response when $b_i=1$ and when $b_i=0$. For $b_i=1$ we have 
\begin{align*}
    (r_i \mid b_i=1) \equiv (\frac{1}{L}\cdot \sum_{F \in U_i} F(x_i) \mid b_i=1)
\end{align*}
For $b_i=0$ we have
\begin{align*}
    (r_i \mid b_i=0)\equiv (\frac{1}{L}\cdot \sum_{F \in U} F(x_i) \mid b_i=0)
\end{align*}

Now since $U_i\mid b_i=1$ and $U\mid b_i=0$ are distributed identically, the summation of the query points are also identically distributed. Therefore, $r_i \mid b_i=0$ and $r_i \mid b_i=1$ are identically distributed. Note that it is crucial that the adversary only queries the point $x_i$ as otherwise we had to take the summation over $U\mid b_i=1$ and $U\mid b_i=0$ which are not identically distributed (the case of $b_i=1$ could have $x_i$ in the training set of the $L$ models).

Since we prove that $r_i\mid b_i=1$ and $r_i\mid b_i=0$ are identical, the adversary cannot distinguish them and the success probability of the adversary is exactly $0.5$. 
The intuitive explanation for this proof is that for each data point, the distribution of output of this algorithm on a given point $x$ is independent of the presence of $x$ in the training set, as we will not use models that are trained with $x$ to answer queries, even if $x$ is in the training set.

\end{proof}
\begin{remark}[A stronger security game and theorem]
Note that there is a worst-case variant of Definition \ref{secGame} where in step 4, instead of the challenger, the adversary select $i\in [2n]$. This is a stronger security game as the adversary can select the worst example in the dataset. However, Theorem \ref{thm:single_query_split_train} remain unchanged in this game. This is because the proof applies to any $i\in[2n]$ and does not require $i$ to be chosen at random. As we will see below, we have another theorem (Theorem \ref{thm:distill}) that considers the privacy of end-to-end \sysname for which the guarantee only holds for the weaker definition. 
\label{remark:strongergame}
\end{remark}

\subsection{\sysname's Privacy under Direct Single-query Attacks}
\label{appendix:selenadsq}
\newcommand{\cM}{\mathcal{M}}
\begin{definition}[stable distillation]
A distillation algorithm $Q\colon M_s\times \mathrm{AUX}\to M_o$ is a potentially randomized algorithm with access to a source model $m_s\in M_s\subseteq Y^X$ and some auxiliary information and returns an  output model $m_o\in M_o \subset Y^X$.  
We define the notion of stability for a distillation algorithm on a point $x\in X$, and joint distribution $\cM$ on $M_s\times \mathrm{AUX}$ as follows:
 $$\mathsf{stablity}(Q,\cM,x)=1-TV(Q(\cM)[x],\cM[x]).$$
Moreover, we say the algorithm $Q$ has $(\alpha,\beta)$-stability on a distribution $\cM$ and a dataset $X$ iff
    $$\Pr_{x\gets X}[\mathsf{stability}(Q,\cM,x)\leq 1-\alpha ]\leq \beta$$
\label{def:stabledistillation}
\end{definition}

\paragraph{Example.} If the distillation algorithm $Q$ ensures that for a specific point $x$ and for all $m_s\in M_s$ we have $Q(m_s)[x] = m_s[x]$, then $Q$ has stability $1$ on point $x$ for all distributions $\cM$ defined on $M_s$.
\begin{remark}
The distillation algorithm $Q$ could also depend on an additional dataset that is correlated with $m_s$ as the auxiliary information. For instance, in our self-distillation algorithm, the distillation is done through the same training set that was used to train $m_s$. In this case, we are interested in the joint distribution $\cM$ that consist of a model $m_s$ as first element and a dataset $D$ as the second element, so that $m_s$ is a model trained on dataset $D$.
\end{remark}
Now we state a corollary of our Theorem \ref{thm:single_query_split_train} about the privacy of the distilled models from the output of the \archname operation. 
\newtheorem{corollary}[theorem]{Corollary}
\paragraph{Notation.} For a learner $C$ and a dataset $X$, we use $\cM_{C,X}$ to denote a distribution of models that is obtained from the following process: First select a random subset $S$ of size $|X|/2$ and then train a model $m$ on that subset using learner $C$ and output $(m,S)$. For a learner $C$ and a distillation model $Q$, we use $QoC$ to denote a learner that first uses $C$ to train a model and then uses distillation algorithm $Q$ to distill that model and then returns the distilled model.
\begin{theorem}\label{thm:distill}
Let $C$ be an arbitrary learner. Assume for a set of samples $X$ the distillation algorithm $Q$ has $(\alpha,\beta)$-stability on distribution $\cM_{C,X}$ and dataset $X$.
    Then, for any adversary $A$ we have
$$\mathsf{SQMI}(A,{QoC},X) \leq \mathsf{SQMI}(A,{C},X)  + \alpha + \beta.$$

\end{theorem}
\begin{proof}
Consider an adversary $A$ that given a response $Qo C[x_i]$ on query $x_i\in X$ outputs a  bit $b'_i=A(Qo C(x_i))$. Let $E$ be an event defined on $X$ such that $E(x)=1$ iff $$\mathsf{stability}(Q,\cM_{C,X},x)\geq 1-\alpha.$$  

For a point $x_i$ such that $E(x_i)=1$ 
we have
\begin{align*}&\Pr\big[A(QoC[x_i])=b_i\big]\leq  \Pr\big[QoC[x_i]\neq C[x_i]\big]\\
&~~ +\Pr\big[A(C[x_i]) =b_i \mid C(x_i)=QoC[x_i]\big]\cdot \Pr\big[QoC[x_i]=C[x_i]\big]\\
&\leq  \alpha + \Pr\big[A(C[x_i])=b_i \big] \end{align*}
Therefore, we have
\begin{align*}&\Pr_{x_i\gets X}\big[A(QoC[x_i])=b_i\big]\\
&\leq \Pr_{x_i\gets X}\big[A(QoC[x_i])=b_i] \mid E(x_i)\big]\cdot \Pr_{x_i\gets X}[E(x_i)] + \Pr_{x_i\gets X}[\bar{E}(x_i)]\\
&\leq \Pr_{x_i\gets X}\big[A(QoC[x_i])=b_i \mid E(x_i)\big]\cdot \Pr_{x_i\gets X}[E(x_i)] + \beta\\
&\leq \Big(\Pr_{x_i\gets X}\big[A(C[x_i])=b_i \mid  E[x_i]\big] + \alpha\Big) \cdot \Pr_{x_i\gets X}[E(x_i)] + \beta\\
& \leq \Pr_{x_i\gets X}\big[A(C[x_i])=b_i]\big] + \alpha + \beta\\
&=\mathsf{SQMI}(A,{C},X) + \alpha + \beta.
\end{align*}

\end{proof}
Now we are ready to state a corollary of Theorems~\ref{thm:distill} and Theorem~\ref{thm:single_query_split_train} for the full pipeline of \archname followed by Self-Distillation. The following Corollary directly follows from Theorems~\ref{thm:distill} and Theorem~\ref{thm:single_query_split_train}.
\begin{corollary}\label{cor:last}
    Let $C_{ST}$ be a learner that uses the \archname algorithm \ref{alg:A1}. Also, let $Q_{SD}$ be a distiller that uses self-distillation algorithm. If $Q_{SD}$ is $(\alpha,\beta)$-stable for a dataset $X$ and distribution $\cM_{C_{ST},X}$, then, for any adversary $A$ we have
    
$$\mathsf{SQMI}(A,Q_{SD}oC_{ST},X) \leq 0.5 + \alpha + \beta.$$
\end{corollary}

\subsection{Discussion of \archname and \sysname for Correlated Points}

\label{appendix:correlationdiscuss}

\begin{remark}[How private is \sysname against multi-query attacks?]
The above theoretical analysis of \sysname is only valid for single-query direct attacks. But one might wonder if we can show a similar theory for privacy of \sysname against multi-query attacks. Unfortunately, we cannot prove a result as general as Corollary \ref{cor:last} for multi-query attacks. In fact, there exist some datasets that \sysname cannot obtain provable privacy for. For instance, imagine a dataset that contains two points $(x,0)$ and $(x',1)$ such that $x$ and $x'$ are almost the same points, i.e. $x\approx x'$, yet they are labeled differently in the training set ($x$ is labeled as 0 and $x'$ as 1). In this scenario, we can observe that the adversary can obtain information about membership of $x$ and $x'$, when querying both points. In particular, if only one of $x$ and $x'$ are selected as members, then we expect the result of query on $x$ and $x'$ to be the same and equal to the label of the one that is selected as a member. However, we argue that this lack of privacy for certain datasets will not manifest in the real world examples as such high correlation does not frequently appear in real-world datasets. Our empirical analysis of \sysname is consistent with this claim. We defer the theoretical analysis of \sysname for multi-query attacks on datasets that satisfy certain assumptions to future work.
\label{remark}
\end{remark}

\paragraph{Specific study of possible leakage in Remark 7.} To study the possible leakage in Remark~\ref{remark} on \archname, we investigate the effect of querying correlated points. In particular, we consider pairs $(x, x')$, where $x$ is a member and $x'$ is a close non-member. Then, we measure the difference between outputs from $L$ sub-models in $Id_{non}(x)$ and random $L$ sub-models for a non-member sample $x'$. This way, we obtain an attack which shows the magnitude of the privacy loss due to the leakage described in Remark~\ref{remark}.

\paragraph{Experiment setup.} We design the following experiment on the CIFAR100 dataset. We use $L_2$ distance to measure the correlation between member samples and non-member samples. For each training sample $x$, we find the sample $x'$ among test set which has the least $L_2$ distance to $x$ but labeled differently. For each  correlated pair $(x, x')$, we query \archname on $x'$ twice, the first query uses $L$ sub-model indices defined by $Id_{non}(x)$ and the second query uses random $L$ sub-models. We denote these two queries by $F_{\theta_\mathrm{I}}(x', Id_{non}(x))$ and $F_{\theta_\mathrm{I}}(x', rnd)$ respectively. Now we can leverage the MIAs evaluated in Section~\ref{sec: eval}: consider $F_{\theta_\mathrm{I}}(x', Id_{non}(x))$ as a member and $F_{\theta_\mathrm{I}}(x', rnd)$ as a non-member, we use these predictions along with the label of $x'$ as input to the direct-single query attacks~(due to their strong performance on undefended models).

\begin{figure}[ht]
\centering
\includegraphics[width=3in]{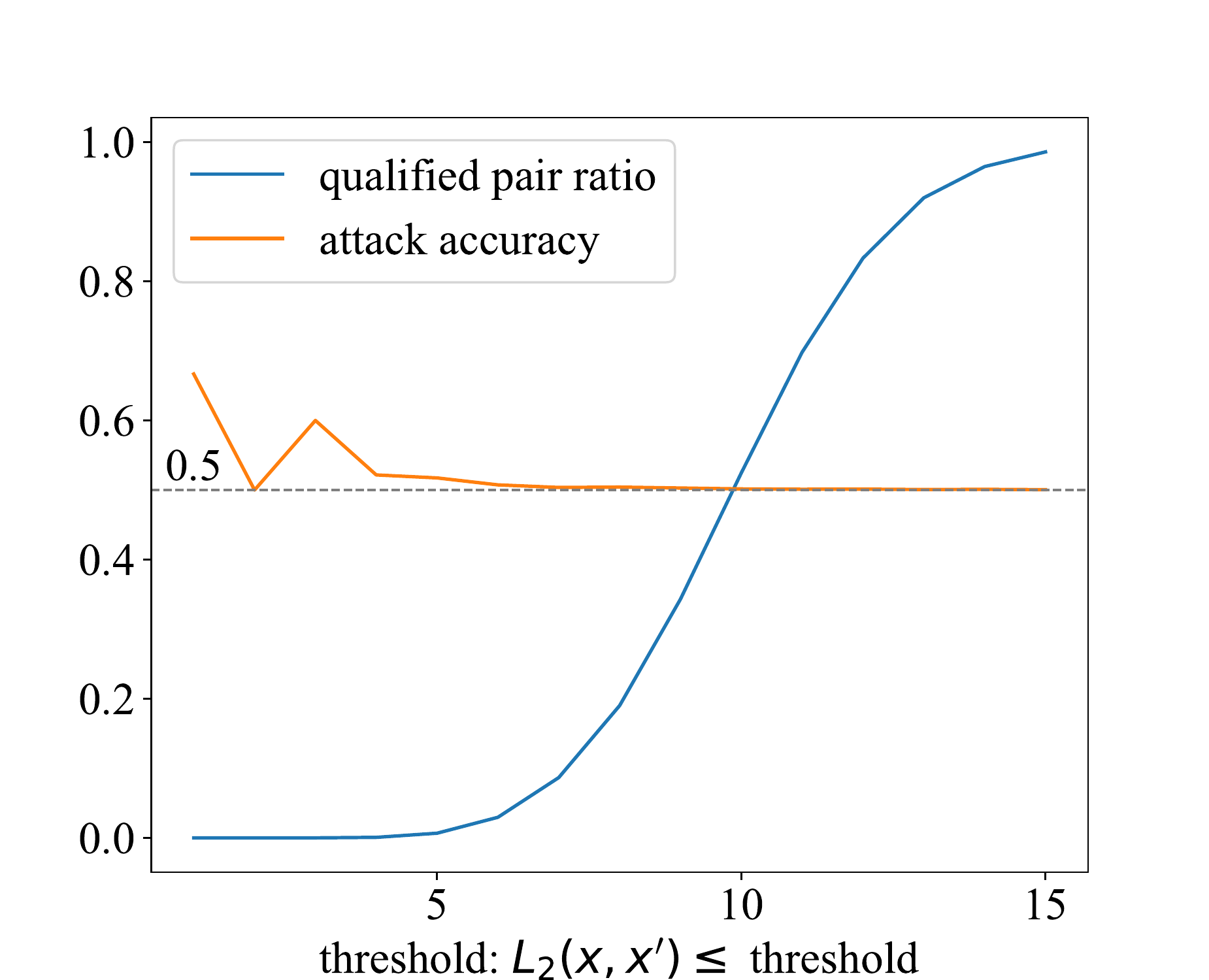}
\caption{Given the $L_2$ distance threshold for correlated pairs $(x, x')$ as the x-axis, we plot the fraction of pairs with distance less than that threshold. We also plot the average attack accuracy among paired queries for that distance.}
\label{fig:ratioandattack}
\end{figure}

\paragraph{Result.} We present the result of the correlated point attack as a function of how close these correlated pairs are, i.e., the distance $L_2(x, x')$. For $L_2$ distance from 1 to 15,\footnote{Image pixel in range [0,1].} we evaluate the ratio of member samples that satisfy this $L_2$ restriction and the attack success rate and plot the result in Figure~\ref{fig:ratioandattack}. We can see that for $L_2$ distance larger than 6, the attack performance is close to a random guess. For $L_2$ distance less than 6, we can see that as $L_2$ distance restriction becomes smaller, the attack accuracy tends to increase. This is consistent with the what we discuss in Remark~\ref{remark}. However, we should also note that the ratio of such pairs that satisfy the restriction is close to 0. Specifically, for $L_2=1$. there are only 6 member samples out of 50000 member samples that satisfy this restriction, which is consistent with our discussion in Remark~\ref{remark} that the presence of such highly correlated pairs in real-world datasets is small.

\paragraph{Can our NN-based attacks (in Section~\ref{sec:attackdefense} and Section~\ref{sec:adaptiveattack}) leverage the correlation leakage?} We emphasize that our NN-based attacks described in Section \ref{sec:attackdefense} and Section~\ref{sec:adaptiveattack} have all the required information for leveraging the correlation leakage described in this subsection. Our attacks have access to a large fraction of dataset together with their membership information and the prediction vector on the target model. Therefore, in principle, the NN-based attack could learn to perform the following: 1) On a given point $x$ find the most correlated point $x'$ in the provided dataset 2) calculate the expected prediction vector for querying $x'$ on models and non-models of $x$.
3) Run the attack described above in the subsection. We cannot prove that the neural network does all these steps, but it has all the power to do so.

\section{Evaluation of Our Architecture's Two Components}
\label{appendix:components}
In Section~\ref{sec: eval}, we have evaluated the end-to-end performance of our framework. Since our framework is a composition of two components, we next evaluate these components individually to study their properties. We first analyze the utility and defense performance of the \archname in isolation, showing that the component achieves test accuracy similar to conventional models while mitigating direct single-query attacks. Second, we demonstrate the necessity of adaptive inference in \archname design, and show that changing our design choices with a baseline approach results in sub-optimal performance. Third, we demonstrate the necessity of the Self-Distillation component, by showing that the \archname component alone is vulnerable against indirect attacks but combined framework provides a strong defense.

    \subsection{Performance of \archname}
    \label{appendix:splitai}
    \textbf{\archname $F_{\theta_\mathrm{I}}$ has similar test accuracy compared with the undefended model and reduces the accuracy of direct single-query attack close to a random guess}. As we discussed via a proof of privacy based on a security game in Appendix~\ref{appendix:securityproof},  \archname alone should mitigate direct single-query attacks (reducing the attack accuracy close to a random guess). Here we experimentally evaluate this property of  $F_{\theta_{\mathrm{I}}}$. Table~\ref{tab:SplitEnsembleF} shows that \archname mechanism successfully mitigates the direct single-query attacks discussed in Section~\ref{subsec: attack} with approximately $50\%$ attack accuracy that is close to a random guess. From a utility perspective, for the worst case, i.e., Purchase100, the test accuracy of $F_{\theta_{\mathrm{I}}}$ is just 0.2\% lower than that the undefended model, which is negligible. For Texas100, the test accuracy of $F_{\theta_{\mathrm{I}}}$ is even $3.5\%$ higher as the ensemble benefits from the average of $L$ models. 

    \begin{table}[ht]
        \caption{Comparison of \archname and undefended model against direct single-query attacks.}
        \label{tab:SplitEnsembleF}
        \centering
        \begin{tabular}{ccccc}
            \toprule
            dataset &defense  &\tabincell{c}{acc on \\training\\ set} &\tabincell{c}{acc on\\test set} & \tabincell{c}{best\\attack}\\
            \midrule
            \multirow{2}{*}{\tabincell{c}{Purchase100}} &No &99.98\% &83.2\% &67.3\% \\
            &$F_{\theta_\mathrm{I}}$ &82.6\% &83.0\%  &{50.3\%}\\
            \midrule
            \multirow{2}{*}{Texas100} &No &79.3\% &52.3\% &66.0\%\\
            &$F_{\theta_\mathrm{I}}$ &56.1\% & 55.9\% &50.7\%\\
            \midrule
            \multirow{2}{*}{\tabincell{c}{CIFAR100}}   &No &99.98\% &77.0\% &74.8\%\\
            &$F_{\theta_\mathrm{I}}$&77.9\% &77.7\% &50.8\% \\
            \bottomrule
        \end{tabular}
    \end{table}

    \subsection{Necessity of Adaptive Inference in \archname}
    \label{subsub: adapen}
    Next, we evaluate the necessity of  our adaptive inference strategy in \archname. To do so, we make a design change in the ensemble that represents a baseline approach for comparison: we apply the strategy of averaging the outputs of all $K$ models for all input samples. We evaluate this choice on the three datasets against the direct single-query attack with $K=25$, $L=10$.

    \begin{table}[ht]
    \caption{Comparison of adaptive inference~(AI) and average of all outputs~(AOAO) strategy against direct single-query attack.}
    \label{tab:adaptiveensemble}
    \centering
    \begin{tabular}{ccccc}
        \toprule
        dataset&ensemble & \tabincell{c}{acc on \\training \\set} & \tabincell{c}{acc on \\test set}& \tabincell{c}{best \\attack}\\
        \midrule
        \multirow{2}{*}{\tabincell{c}{Purchase\\100}} & AI &82.6\% & 83.0\% &50.3\%\\
        & AOAO &99.9\% &83.4\% &62.0\%\\
        \midrule
        \multirow{2}{*}{Texas100}&AI &56.1\% &55.9\% &50.7\%\\
        &AOAO &81.9\% &56.6\% &67.7\%\\
        \midrule
        \multirow{2}{*}{CIFAR100} &AI &77.9\% &77.7\% &50.8\%\\
        &AOAO &99.98\% &78.1\% &69.2\%\\
        \bottomrule
    \end{tabular}
    \end{table}    
    
    Table~\ref{tab:adaptiveensemble} presents the comparison of adaptive inference with the baseline strategy of averaging all sub-model outputs. \textbf{While the adaptive inference approach reduces the direct single-query attack to a random guess, the behavior for average of all outputs on members and non-members is very different.} For example, a generalization gap still exists in average of all outputs: 16.5\% on Purchase100, 25.3\% on Texas100 and 21.88\% on CIFAR100. The best attack accuracy against average of all outputs is higher than 60.0\% on all three datasets, which still indicates a severe membership inference threat.
    
    The adaptive inference is needed in \archname to achieve a good defense against direct single-query attack. Since Self-Distillation needs to transfer knowledge from a source model which has a good defensive abilities, the whole \archname is a key component in our whole system.

    \subsection{Necessity of Self-Distillation}
    \label{appendix_subsec:necessityofdistill}
    We have stated the need to introduce Self-Distillation as a second component to overcome weaknesses of \archname $F_{\theta_\mathrm{I}}$ in Section~\ref{subsec:sd}. 
    We now demonstrate this by showing the potential membership inference risks in $F_{\theta_\mathrm{I}}$ are mitigated by our final protected model from Self-Distillation $F_{\theta_{\mathrm{II}}}$. 
    Towards this end, we now focus on indirect single-query attack, in which the attacker adds noise to the target sample and queries the noisy sample. We generate the noisy sample by randomly flipping one feature for binary inputs and randomly increasing or decreasing one pixel by 1 for images. Results are presented in Table~\ref{tab:noiseindirect}.

    \begin{table}[ht]
    \caption{Comparison for \archname $F_{\theta_\mathrm{I}}$ and \sysname $F_{\theta_{\mathrm{II}}}$ against indirect single-query attack.}
    \centering
    \label{tab:noiseindirect}
        \begin{tabular}{cccccc}
        \toprule
        dataset &model &\tabincell{c}{noisy\\ data?}  &\tabincell{c}{acc on \\training\\ set} &\tabincell{c}{acc on \\test set} &  \tabincell{c}{best \\ attack}\\
        \midrule
        \multirow{4}{*}{\tabincell{c}{Purchase\\100}} &\multirow{2}{*}{$F_{\theta_\mathrm{I}}$} &no &82.6\% &83.0\%  & 50.3\%\\
         & &yes &99.1\% & 83.0\% & 60.8\% \\
        &\multirow{2}{*}{$F_{\theta_{\mathrm{II}}}$}&no &82.7\% &79.3\%  & {53.3\%}\\
        &  &yes &82.2\% & 79.3\% & 53.1\% \\
        \midrule
        \multirow{4}{*}{\tabincell{c}{Texas100}} &\multirow{2}{*}{$F_{\theta_\mathrm{I}}$} &no  &56.1\% & 55.9\% &{50.7\%}\\
         & & yes & 79.3\% & 55.9\% &64.1\% \\
        &\multirow{2}{*}{$F_{\theta_{\mathrm{II}}}$} &no  &58.8\% & 52.6\% & {54.8\%} \\
         & & yes & 57.9\% & 52.6\% &53.2\%\\
        \midrule
        \multirow{4}{*}{\tabincell{c}{CIFAR100}} &\multirow{2}{*}{$F_{\theta_\mathrm{I}}$} &no &77.9\% &77.7\% &50.8\% \\
         & & yes &99.6\% &78.3\% &66.0\%\\
        &\multirow{2}{*}{$F_{\theta_{\mathrm{II}}}$} &no &78.1\% &74.6\% &55.1\% \\
        & & yes & 78.1\% &74.6\% &55.0\% \\
        \bottomrule
        \end{tabular}
    \end{table}

    While the indirect single-query attack success against $F_{\theta_{\mathrm{I}}}$ is high: 60.8\% for Purchase100, 64.1\% for Texas100 and 66.0\% for CIFAR100, the membership privacy for $F_{\theta_{\mathrm{II}}}$ does not degrade using such indirect attacks. 
    
    Recall that Self-Distillation is done by using only one exact query for each training sample and applying conventional training with resulting soft labels to train the protected model $F_{\theta_{\mathrm{II}}}$. Thus \textbf{the MIA success against the protected model $F_{\theta_{\mathrm{II}}}$ (from Self-Distillation) using noisy data is not higher than using clean data}. Therefore, Self-Distillation solves a key privacy challenge faced by \archname and is a necessary component in our framework.
    Furthermore, the Self-Distillation approach also solves the issues of computational overhead in the inference stage and the replay attack as discussed in Section~\ref{subsec:sd}.

    \section{Label-Only Attacks}
    \label{appendix:label}
    We analyze boundary attack for all datasets and data augmentation attack for CIFAR100.

    \emph{Boundary attack:} we use CW white-box attack~\cite{carlini2017towards} for computer vision dataset (CIFAR100 in our paper) as other black-box attacks need to make multiple queries on the target model and can not achieve better performance than CW white-box attacks, as shown by Choo et. al.~\cite{choo2020label}. We have 
    $$I_{\text{CW}}(F, \textbf{x}, y)= \mathds{1}\{adv-{dist}_{CW}(\textbf{x}) \geq {\tau_{(y)}}\}$$
    For other binary datasets considered in our work, the CW attack is not very successful due to the binary features. The only possible feature values are $0/1$, thus successful adversarial examples require turning feature value of 1 to be lower than 0.5 or turning feature value of 0 to be higher than 0.5, which is a big jump, otherwise the rounded noisy sample is likely to be the same as the target sample. Instead, we introduce noise in the target sample by randomly flipping a threshold number of features~\cite{choo2020label, li2020label}. Given a threshold on the number of flipped features, we generate hundreds of noisy samples for each target sample to query the model. We then perform an attack based on the percentage of correct predictions on the noisy samples to estimate the boundary. This is based on the intuition that for samples which are away from the classification boundary, the samples around it are likely to be correctly classified. Hence the metric of correctness percentage on noisy samples can be used to estimate the boundary distance. We vary the number of flipped features from 1 to 30 for Purchase100 and from 1 to 100 for Texas100~(
    We find that our selected parameters already provide a search space large enough for the optimal threshold because continuing to  increase the threshold will have lower attack accuracy as the members become more noisy). 

    We report the best attack accuracy among these numbers.
    $$I_{\text{random-noise}}(F, \textbf{x}, y)= \mathds{1}\{\frac{\sum_{\textbf{x}' \text{around \textbf{x}}} corr(\textbf{x}')}{|\textbf{x}' \text{around \textbf{x}}|} \geq {\tau_{(y)}}\}$$
    
    \emph{Data Augmentation Attack:} Data augmentation attack is based the augmentation technique that we use to train the model. During training, we first perform image padding and cropping, and then perform horizontal flipping with a $0.5$ probability to augment the training set. An attacker will similarly query all possible augmented results of a target image sample. For example, if the padding size is 4 for left/right and up/down, and the size of cropped image is the same as original image: considering left/right after cropping, there are $(4+4+1)$ possible choices; considering up/down after cropping,  there are $(4+4+1)$ possible choices; considering horizontal flipping, there are 2 possible choices. Therefore, the number of total queries for a target image is $9\times9\times2 = 162$. As the target model is more likely to correctly classify the augmented samples of members than that of non-members, only target samples with sufficient correctly classified queries will be identified as members. This attack is important as Choo et al.~\cite{choo2020label} show that data augmentation attacks~(specifically those based on image translation) can achieve higher performance than CW attacks. We have 
    $$I_{\text{data-augmentation}}(F, \textbf{x}, y)= \mathds{1}\{\frac{ \sum_{\textbf{x}': \text{augmented \textbf{x}}} corr(\textbf{x}')}{|\textbf{x}': \text{augmented \textbf{x}}|} \geq {\tau_{(y)}}\}$$

    \begin{table*}[ht]
    \caption{Ablation study on architectures for Purchase100. The first column describes model architecture in format of (activation function, width, depth). AdvReg refers to adversarial regularization. The last column is the highest attack accuracy for each row, i.e. for a specific defense on one dataset, the highest attack accuracy that MIAs can achieve, which gives an overview of comparison: the lower the best attack accuracy, lower the membership inference threat. For each dataset, the defense which has the lowest corresponding attack accuracy is bold in the column of best direct single-query attack, best label-only and best attack.}
    \label{tab:ablationpurchase}
    \centering
    \begin{tabular}{cccccccc}
        \toprule
        \tabincell{c}{architectures\\
        (activation function, \\width, depth)}&defense &\tabincell{c}{acc on \\training set} &\tabincell{c}{acc on \\test set} &\tabincell{c}{best direct \\single-query\\ attack}&\tabincell{c}{best\\label-only \\attack}&\tabincell{c}{best adaptive \\attack} &\tabincell{c}{best attack}\\
        \midrule
        \multirow{4}{*}{Tanh, 1, 4} &None &99.98\% &83.2\% &{67.3\%} &65.8\%  &N/A &67.3\%\\
        &MemGuard &99.98\% &83.2\% &58.7\% &{65.8\%} &N/A &65.8\%\\
        &{AdvReg} &91.9\% &78.5\% &57.3\% &{57.4\%} &N/A &57.4\%\\
        &\textbf{\sysname} &82.7\% &79.3\% &\textbf{53.3\%} &\textbf{53.2\%} &{54.3\%} &\textbf{54.3\%}\\
        \midrule
        \multirow{4}{*}{Tanh, 1, 3} &None &100.0\% &84.9\% &{68.2\%} &65.9\% &N/A &68.2\%\\
        &MemGuard &100.0\% &84.9\% &57.6\% &{65.9\%} &N/A &65.9\%\\
        &{AdvReg} &89.2\% &78.2\% &{56.6\%} &56.5\% &N/A &56.6\%\\
        &\textbf{\sysname} &83.9\% &81.1\% &\textbf{52.5\%} &\textbf{52.6\%} &{53.4\%} &\textbf{53.4\%}\\
        \midrule
        \multirow{4}{*}{Tanh, 1, 5} &None &99.8\% &81.4\% &{66.7\%} &65.7\% &N/A &66.7\%\\
        &MemGuard &99.8\% &81.4\% &59.5\% &{65.7\%} &N/A &65.7\% \\
        &{AdvReg} &91.7\% &77.3\% &58.2\% &{58.4\%} &N/A &58.4\%\\
        &\textbf{\sysname} &82.6\% &78.8\% &\textbf{54.5\%} &\textbf{54.9\%} & {56.2\%}  &\textbf{56.2\%}\\ 
        \midrule
        \multirow{4}{*}{Tanh, 0.5, 4} &None &99.9\% &79.9\% &{67.9\%} &66.7\% &N/A &67.9\%\\
        &MemGuard &99.9\% &79.9\% &60.2\% &{66.7\%} &N/A &66.7\%\\
        &{AdvReg} &92.8\% &77.6\% &58.8\% &{58.9\%}&N/A &58.9\%\\
        &\textbf{\sysname}  &82.5\% &77.8\% &\textbf{53.7\%} &\textbf{53.6\%} &{55.0\%} &\textbf{55.0\%}\\
        \midrule
        \multirow{4}{*}{Tanh, 2, 4} &None &100.0\% &84.4\% &{70.7\%} &67.6\% &N/A &{70.7\%}\\
        &MemGuard &100.0\% &84.4\% &58.7\% &{67.6\%} &N/A &67.6\%\\
        &{AdvReg} &90.7\% &77.6\% &57.1\% &{57.2\%}&N/A &57.2\%\\
        &\textbf{\sysname} &83.6\% &80.5\% &\textbf{54.5\%} &\textbf{54.9\%} &{56.0\%} &\textbf{56.0\%}\\
        \midrule
        \multirow{4}{*}{ReLU, 1, 4} &None &99.2\% &79.7\% &{63.6\%} &63.3\% &N/A &63.6\%\\
        &MemGuard &99.2\% &79.7\% &59.4\% &{63.3\%} &N/A &63.3\%\\
        &{AdvReg} &92.4\% &76.8\% &58.4\% &{58.6\%} &N/A &58.6\%\\
        &\textbf{\sysname} &82.5\% &77.7\% &\textbf{{54.3\%}} &\textbf{53.9\%} &53.7\% &\textbf{54.3\%}\\
    \bottomrule
    \end{tabular}
\end{table*}

\section{Experiment Setup}
Here we introduce the datasets,  the model architectures, and the hyper-parameter settings in more detail. 
\label{appendix:setup}
\subsection{Dataset}
We use three benchmark datasets widely used in prior works on MIAs:

\textbf{CIFAR100}: This is a benchmark dataset used to evaluate image classification algorithms~\cite{krizhevsky2009learning}. CIFAR100 is composed of $32\times 32$ color images in 100 classes, with 600 images per class. For each class label, 500 images are used as training samples, and remaining 100 images are used as test samples.

    \textbf{Purchase100}: This dataset is based on Kaggle's Acquire Valued Shopper Challenge,\footnote{\href{https://www.kaggle.com/c/acquire-valued-shoppers-challenge}{https://www.kaggle.com/c/acquire-valued-shoppers-challenge}.} which contains shopping records of several thousand individuals. We obtained a prepossessed and simplified version provided by Shokri et al.~\cite{shokri2017membership}. This dataset is composed of 197,324 data samples with 600 binary features. Each feature corresponds to a product and represents whether the individual has purchased it or not. This dataset is clustered into 100 classes corresponding to purchase styles.

    \textbf{Texas100}: This dataset is based on the Hospital Discharge Data public use files with information about inpatients stays in several health facilities released by the Texas
    Department of State Health Services from 2006 to 2009.\footnote{\href{https://www.dshs.texas.gov/THCIC/Hospitals/Download.shtm}{https://www.dshs.texas.gov/THCIC/Hospitals/Download.shtm}.} Each data record contains external causes of injury, the diagnosis, the procedures the patient underwent and some generic information. We obtain a prepossessed and simplified version of this dataset provided by Shokri et al.~\cite{shokri2017membership}, which is composed of 67,330 data samples with 6,170 binary features. This dataset is used to classify 100 most frequent used procedures.

    \subsection{Target Models}
    \label{sub: target}
    For CIFAR100, we use ResNet-18~\cite{he2016deep}, which is a benchmark machine learning model widely used in computer vision tasks. We adopt the cross-entropy loss function and use Stochastic Gradient Descent~(SGD) to learn the model parameters. We train the model for 200 epochs with batch size of 256, initializing learning rate 0.1 with weight decay 0.0005 and Nesterov momentum of 0.9 and divide the learning rate by 5 at epoch 60, 120, 160.\footnote{\href{https://github.com/weiaicunzai/pytorch-cifar100}{https://github.com/weiaicunzai/pytorch-cifar100}.}
    
    For Purchase100 and Texas100, we follow previous work~\cite{nasr2018machine} to use a 4-layer fully connected neural network with layer sizes $[1024,512,$ $256,100]$ and Tanh as the activation function. We use the cross-entropy loss function and Adam~\cite{kingma2014adam} optimizer to train the model on Purchase100 for 30 epochs and on Texas100 for 20 epochs with learning rate of 0.001. The batch size is 512 for Purchase100 and 128 for Texas100. 
    
    These hyper-parameter settings are also used for our ablation study in Appendix~\ref{appendix:architectures} where we vary the model architecture.

\begin{table*}[ht]
    \caption{Ablation study on architectures for Texas100. The first column describes model architecture in format of (activation function, width, depth). AdvReg refers to adversarial regularization. The last column is the highest attack accuracy for each row, i.e. for a specific defense on one dataset, the highest attack accuracy that MIAs can achieve, which gives an overview of comparison: the lower the best attack accuracy, lower the membership inference threat. For each dataset, the defense which has the lowest corresponding attack accuracy is bold in the column of best direct single-query attack, best label-only and best attack.}
    \label{tab:ablationtexas}
    \centering
    \begin{tabular}{cccccccc}
        \toprule
        \tabincell{c}{architectures\\
        (activation function, \\width, depth)}&defense &\tabincell{c}{acc on\\ training set} &\tabincell{c}{acc on\\ test set} &\tabincell{c}{best direct \\single-query\\ attack}&\tabincell{c}{best\\label-only \\attack}&\tabincell{c}{best adaptive \\attack} &\tabincell{c}{best attack}\\
        \midrule
        \multirow{4}{*}{Tanh, 1, 4} &None &79.3\% &52.3\% &{66.0\%} &64.7\% &N/A &66.0\% \\
        &MemGuard &79.3\% &52.3\% &63.0\% &{64.7\%} &N/A &64.7\% \\
        &{AdvReg} &55.8\% &45.6\% &{60.5\%} &56.6\% &N/A &60.5\%\\
        &\textbf{\sysname} &58.8\% &52.6\% &\textbf{54.8\%} &\textbf{55.1\%}&54.9\% &\textbf{55.1\%}\\
        \midrule
        \multirow{4}{*}{Tanh, 1, 3} &None &82.1\% &55.5\% &{66.2\%} &65.5\% &N/A  &66.2\%\\
        &MemGuard &82.1\% &55.5\% &63.4\% &{65.5\%}&N/A &65.5\%\\
        &{AdvReg} &54.9\% &47.0\% &{58.6\%} &55.5\%  &N/A &58.6\%\\
        &\textbf{\sysname} &61.1\% &55.4\% &\textbf{54.5\%} &\textbf{54.6\%} &{55.6\%} &\textbf{55.6\%}\\
        \midrule
        \multirow{4}{*}{Tanh, 1, 5} &None &76.5\% &49.0\% &{66.4\%} &65.2\% &N/A &66.4\%\\
        &MemGuard &76.5\% &49.0\% &63.9\% &{65.2\%} &N/A &65.2\%\\
        &{AdvReg} &54.4\% &43.4\% &{61.5\%} &56.9\% &N/A &61.5\%\\
        &\textbf{\sysname} &56.7\% &51.3\% &\textbf{54.3\%} &\textbf{53.9\%} &52.9\% &\textbf{54.3\%}\\        
        \midrule
        \multirow{4}{*}{Tanh, 0.5, 4} &None &76.5\% &52.1\% &{65.9\%}&63.2\%  &N/A &65.9\%\\
        &MemGuard &76.5\% &52.1\% &63.0\% &{63.2\%} &N/A &63.2\%\\
        &{AdvReg}  &57.1\% &45.4\% &{62.2\%}&57.1\% &N/A &62.2\%\\
        &\textbf{\sysname} &57.8\% &53.1\% &\textbf{53.7\%} &\textbf{54.0\%}&{54.7\%} &\textbf{54.7\%}\\
        \midrule
        \multirow{4}{*}{Tanh, 2, 4} &None &81.7\% &51.9\% &{67.7\%} &67.0\% &N/A &67.7\%\\
        &MemGuard &81.7\% &51.9\% &64.7\% &{67.0\%} &N/A &67.0\%\\
        &{AdvReg} &52.6\% &44.9\% &{57.9\%}&54.7\% &N/A &57.9\%\\
        &\textbf{\sysname} & 59.7\% &53.6\% &\textbf{54.4\%} &\textbf{54.7\%} &53.6\%  &\textbf{54.7\%}\\
        \midrule
        \multirow{4}{*}{ReLU, 1, 4} &None &98.8\% &47.0\% &{81.7\%}&80.7\% &N/A &81.7\%\\
        &MemGuard &98.8\% &47.0\% &75.8\% &{80.7\%} &N/A &80.7\%\\        
        &{AdvReg} &55.0\% &43.0\% &{59.0\%} &57.0\% &N/A &59.0\%\\
        &\textbf{\sysname} &54.6\% &51.4\% &\textbf{57.6\%}&\textbf{54.4\%} &55.7\%  &\textbf{57.6\%}\\
    \bottomrule
    \end{tabular}
\end{table*}

\section{Ablation Studies}
\label{appendix:ablation}

In this section, we first report on ablation studies that vary the model architecture to demonstrate that the benefits of \sysname hold across architectures. Second, we report on ablation studies that vary our parameters $K$ and $L$ and discuss parameter trade-offs and selection. 

\subsection{Ablation Study on Different Model Architectures}
\label{appendix:architectures}
For Purchase100 and Texas100, the target classifier is a 4-layer fully connected neural network. We test two additional neural network depths by deleting the last hidden layer (depth=3) or adding one more hidden layer with 2048 neurons (depth=5). We test two additional neural network widths by halving the numbers of hidden neurons (width=0.5) or doubling the numbers of hidden neurons (width=2.0). We also test both ReLU and Tanh, as the activation functions. For CIFAR100, we apply  \sysname on two different architectures: ResNet-18~\cite{he2016deep} and VGG-16~\cite{simonyan2014very}. Note that we will optimize the choice of $K$ and $L$ for different model architectures to achieve the best trade-off between test accuracy and membership privacy.

\begin{table*}[ht]
    \caption{Ablation study on architectures for CIFAR100. AdvReg refers to adversarial regularization. The last column is the highest attack accuracy for each row, i.e. for a specific defense on one dataset, the highest attack accuracy that MIAs can achieve, which gives an overview of comparison: the lower the best attack accuracy, lower the membership inference threat. For each dataset, the defense which has the lowest corresponding attack accuracy is bold in the column of best direct single-query attack, best label-only and best attack.}
    \label{tab:ablationcifar}
    \centering
    \begin{tabular}{cccccccc}
        \toprule
        architectures&defense &\tabincell{c}{acc on \\training set} &\tabincell{c}{acc on\\ test set} &\tabincell{c}{best direct \\single-query\\ attack}&\tabincell{c}{best\\label-only \\attack}&\tabincell{c}{best adaptive \\attack} &best attack\\
        \midrule
        \multirow{4}{*}{ResNet-18} &None &99.98\% &77.0\% &{74.8\%} &69.9\% &N/A &74.8\% \\
        &MemGuard &99.98\% &77.0\% &68.7\% &{69.9\%} &N/A &69.9\%\\
        &{AdvReg} &86.9\% &71.5\% &58.6\% &{59.0\%}&N/A &59.0\%\\
        &\textbf{\sysname} &78.1\%  &74.6\% &\textbf{55.1\%}&\textbf{54.0\%} &{58.3\%} &\textbf{58.3\%}\\
        \midrule
        \multirow{5}{*}{VGG-16} &None &99.97\% &74.3\% &71.1\% &{73.6\%}&N/A &73.6\%\\
        &MemGuard &99.97\% &74.3\% &64.8\% &{73.6\%} &N/A &73.6\%\\
        &{AdvReg} &95.4\% &70.3\% &64.9\% &{67.2\%} &N/A &67.2\%\\
        &\textbf{\sysname} &75.3\% &71.1\% &\textbf{54.7\%}&\textbf{55.5\%} & {57.3\%} &\textbf{57.3\%}\\        
    \bottomrule
    \end{tabular}
\end{table*}

    \begin{table*}[ht]
    \caption{\archname and \sysname against direct single-query attack for $K$ = 25 on Purchase100.}
    \label{varyL}
    \centering
    \begin{tabular}{ccccccccccccc}
        \toprule
        $L$ & 5 & 6 &7 &8 &9 &10 &11 &12 &13 & 14 &15\\
        \midrule
        \tabincell{c}{\archname  $F_{\theta_{\mathrm{I}}}$ single \\model accuracy \\on test set} &81.3\% &80.8\% &80.1\% &79.5\% &78.7\% &77.9\% &77.1\% &75.7\% &75.2\% &73.8\% &72.6\%\\
        \midrule
        \tabincell{c}{\archname  $F_{\theta_{\mathrm{I}}}$ acc \\on test set} &83.8\% &83.9\% &83.5\% &83.5\% &83.1\% &83.0\% &82.7\% &82.3\% &82.0\% &81.5\% &81.0\%\\
        \midrule
        \tabincell{c}{best direct \\single-query \\attack against  $F_{\theta_{\mathrm{I}}}$} &50.4\% &50.4\%  &50.6\% &50.5\% &50.5\% &50.3\% &50.7\% &50.7\% &50.3\% &50.7\% &50.6\%\\
        \midrule
        \tabincell{c}{\sysname $F_{\theta_{\mathrm{II}}}$ acc \\on test set}&79.8\% &79.9\% &79.9\% &79.5\% &79.3\% &79.3\% &78.8\% &78.5\% &78.1\% &77.8\% &77.3\%\\
        \midrule
        \tabincell{c}{best direct \\single-query \\attack against  $F_{\theta_{\mathrm{II}}}$}&55.7\% &55.0\% &54.6\% &54.5\% &53.9\% &53.3\% &52.9\% &52.0\% &52.3\% &51.8\% &51.6\%\\
        \bottomrule
    \end{tabular}
    \end{table*}

\begin{table}[ht]
    \caption{\archname and \sysname against direct single-query attack for $K/L$ = 5/2 on Purchase100.}
    \label{fixedKLfull}   
    \centering
    \begin{tabular}{ccccc}
        \toprule
        model&$K$, $L$  &\tabincell{c}{single model\\
        acc on test set}  &\tabincell{c}{acc on \\test set} &\tabincell{c}{best \\attack}\\
        \midrule
        \multirow{3}{*}{\archname}&5, 2 &78.1\% & 80.9 \% &50.5\% \\
        &25, 10 & 77.9\%  &83.0\%& 50.3\%\\
        &50, 20 &77.9\% &83.5\% &50.4\%\\
        \midrule
        \multirow{3}{*}{\sysname} &5, 2 &N/A &77.7\% &55.3\%\\
        &25, 10 &N/A &79.3\% &53.3\%\\
        &50, 20 &N/A &79.2\% &53.2\%\\
        \bottomrule
    \end{tabular}
    \end{table} 
Table~\ref{tab:ablationpurchase} presents the ablation study on Purchase100. Compared with undefended model, \sysname only incurs 2.0\%$\sim$3.9\% loss in (test) classification accuracy. As for MemGuard, though it has the same classification accuracy as undefended model, the best attack accuracy across different architectures are higher than 63.0\%~(MemGuard cannot defend against label-only attacks) while \sysname limits the attack to be no more than 56.2\%. Compared with adversarial regularization, \sysname achieves 0.2\%$\sim$2.9\% higher classification accuracy and reduces the additional attack accuracy over a random guess~(50\%) by a factor of 1.2 $\sim$ 2.

Table~\ref{tab:ablationtexas} presents the ablation study on Texas100.
Note that the classification accuracy of \sysname is only 0.1\% lower than the undefended model for the model which has a architecture of width=1, depth=3 and Tanh as the activation function. For other architectures, \sysname even increases the classification accuracy a little~(0.3\% $\sim$ 4.4\%). As for MemGuard, the best attack accuracy across different architectures is higher than 63.0\%, and as high as 80.7\% for model with width =1, depth=4 and ReLU activation function~(MemGuard cannot defend against label-only attacks), while \sysname limits the attack accuracy to no more than $57.6\%$.
Compared with adversarial regularization,  \sysname achieves higher classification accuracy~(7.0\% $\sim$ 8.7\%) and lower MIA accuracy~(1.4\% $\sim$ 7.5\%).

Table~\ref{tab:ablationcifar} presents the ablation study on CIFAR100. Compared with undefended model, the classification accuracy for \sysname only decreases by 2.4\% $\sim$ 3.2\%. As for MemGuard, the best attack accuracy across different architectures is 69.9\% $\sim$ 73.6\%~, while \sysname limits the attack accuracy no more than $58.3\%$. In comparison with adversarial regularization, \sysname achieves 0.8\% $\sim$ 3.1\% higher classification accuracy and 0.7\% $\sim$ 9.9\% lower MIA accuracy.

\subsection{Ablation Study on K and L}
\label{appendix:KL}
    Here, we discuss the setting of parameters in the \archname and \sysname, i.e., the choice of $K$ and $L$. We first vary $L$ keeping $K$ fixed~(Table~\ref{varyL}) and second vary $K$ and $L$ keeping the ratio of $K/L$ fixed~(Table~\ref{fixedKLfull}). We evaluate these two experiments on Purchase100. We first discuss the performance of \archname. From a privacy perspective, we find that for direct single-query attacks, all settings of $K$ and $L$ in \archname limit the attack accuracy around a random guess. From a utility perspective, when $L$ is smaller with fixed $K$, each sub-model is trained with more data, and the accuracy of a single sub-model on the test set is higher. For overall performance in \archname, which is the average of $L$ outputs, when $L$ is smaller, fewer models are aggregated. For example, for $L=1$, the test accuracy is lower than that of the model trained with whole dataset. When $L$ increases close to $K$, the test accuracy for each single sub-model is low and the overall accuracy will be lower than the test accuracy of the undefended model. Therefore, there is a trade-off in the choice of the parameter $L$.
    
    When the ratio of $K/L$ is fixed, the test accuracy will be lower than the undefended model for small values of $L$ as the ensemble performance of $L$ is poor. As $L$ increases, the test accuracy will increase because of ensemble of $L$ models but such improvement is limited: the test accuracy is nearly same for $K=25$, $L=10$ and $K=50$, $L=20$.

    We next discuss the performance of \sysname. As the final protected model in \sysname learns the knowledge transferred from the \archname, it has also similar performance for $L$ ranging from 5 to 7, and with further increases in $L$, the test accuracy of \sysname also drops. For MIA accuracy, when $K$ is kept fixed, MIA accuracy decreases as L increases. When the ratio of $K/L$ is fixed, for $K=25$, $L=10$ and $K=50$, $L=20$, the MIA accuracy is similar. This is because when $L$ is relative large, it's easier for Self-Distillation to train a model which mimics \archname: when a single sub-model is trained with a appropriate portion of data, it will have a good test accuracy which serves as an enabler for Self-Distillation.

    Keeping computational overhead, utility and privacy in mind, \textbf{we need $L$ to be large enough and a proper proportion of $K$ to ensure that each sub-model is trained with enough data and benefits from averaging while enabling Self-Distillation to mimic the performance of \archname.} However, we should not simply increase $L$ without any limitations while keeping $K/L$ fixed, referring to the computation overhead discussed in Section~\ref{subsec: efficiency}. In our experiments, we use $K=25$, $L = 10$ for all three datasets.

    \section{An Optional Parameter for Soft Labels in Self-Distillation for Trade-off between Utility and Membership Privacy}
    \label{appendix:tradeoff}

    The (test) classification accuracy of  \sysname is 2.0\% $\sim$ 3.9\% lower than undefended model on Purchase100 and 2.4\% $\sim$ 3.2\% lower than undefended model on CIFAR100. We now consider an alternative design choice for soft labels of training set in Self-Distillation: instead of only using output from \archname to optimize for privacy, we now combine the outputs from \archname  and the ground truth label together as soft labels in Self-Distillation. This results in a trade-off which can help achieve higher classification accuracy than \sysname at the cost of  membership privacy:
    \begin{equation}
    \label{eqa:alphaparameter}
        y_{\text{soft}} = (1-\lambda) F_{\theta_{\mathrm{I}}}+\lambda y
    \end{equation}
    Here $\lambda$ is a hyper-parameter ranging from 0 to 1, which controls the ratio of outputs from \archname and ground truth labels for training set. When $\lambda =0$, this is equivalent to \sysname; when $\lambda = 1$, this is equivalent to undefended model.

\begin{figure}[ht]
\centering
\includegraphics[width=3in]{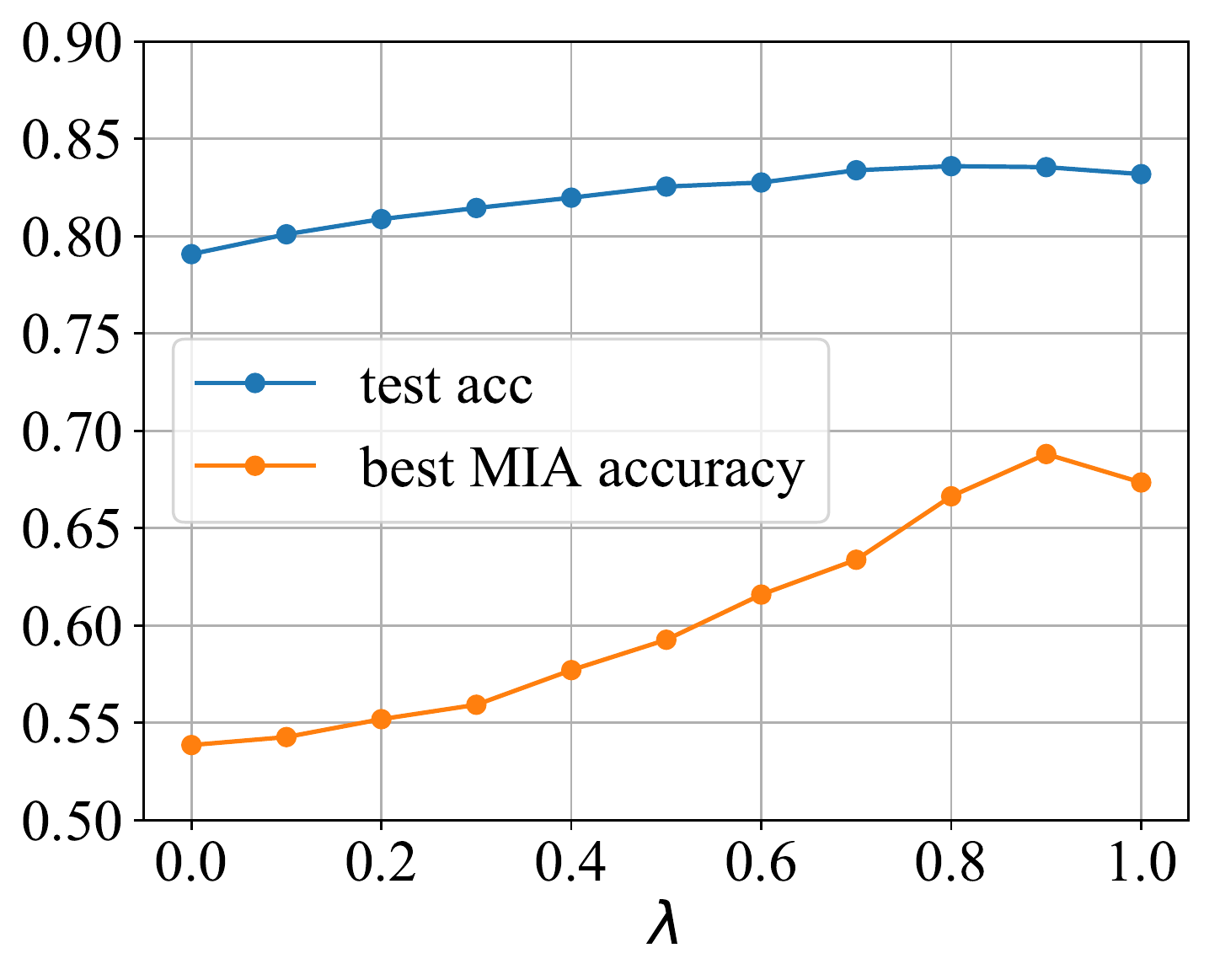}
\caption{The change of classification accuracy with $\lambda$ and change of best MIA classification accuracy with $\lambda$. When $\lambda$ = 0, this
is equivalent to SELENA; when $\lambda$ = 1, this is equivalent to undefended model.}
\label{fig:alphaparamter}
\end{figure}

    \begin{figure*}[ht]
        \centering
        \subfigure[Purchase100]{
        \begin{minipage}[t]{0.3\linewidth}
        \centering
        \includegraphics[width=2.1in]{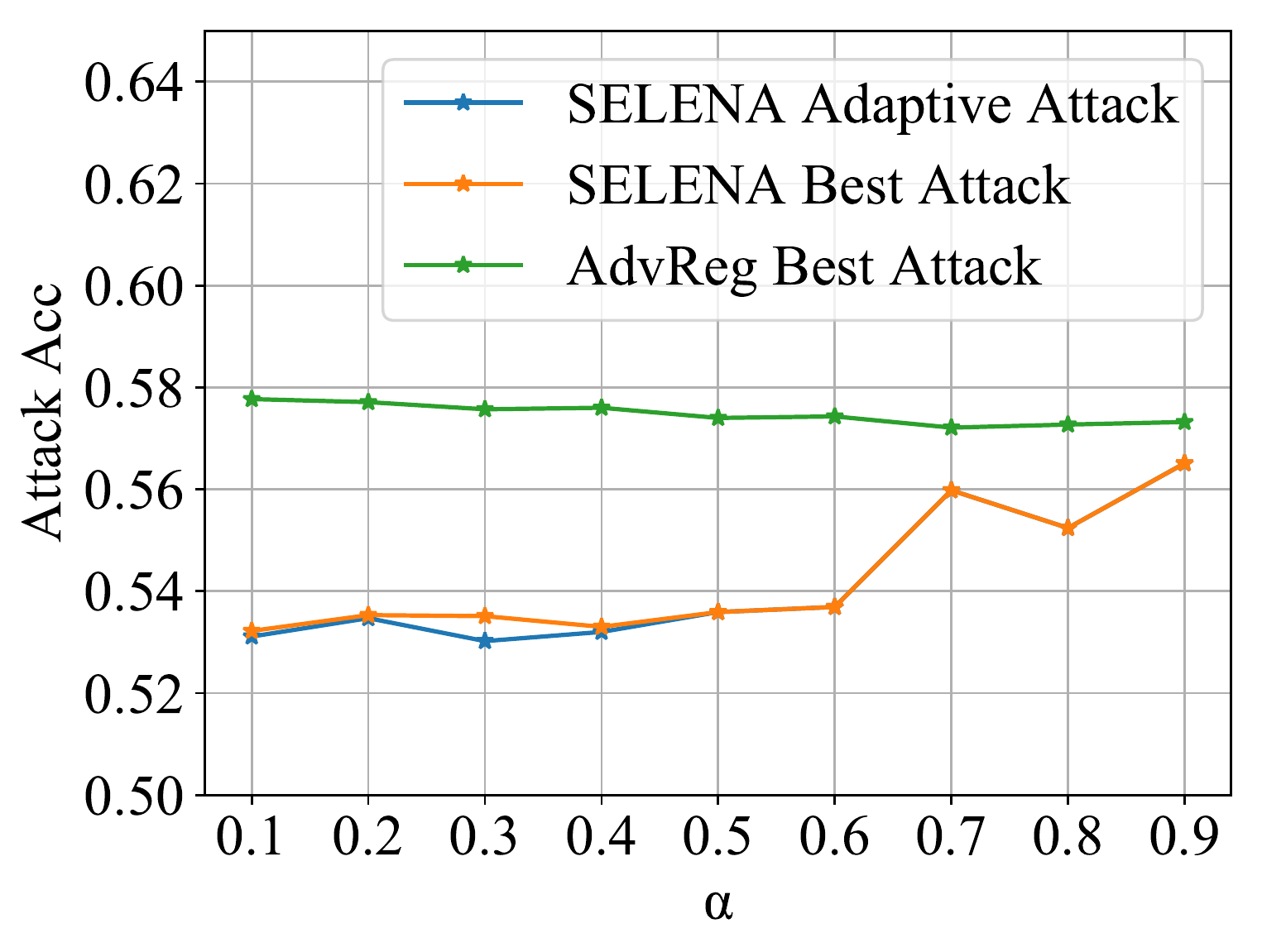}
        \end{minipage}
        }%
        \subfigure[Texas100]{
        \begin{minipage}[t]{0.3\linewidth}
        \centering
        \includegraphics[width=2.1in]{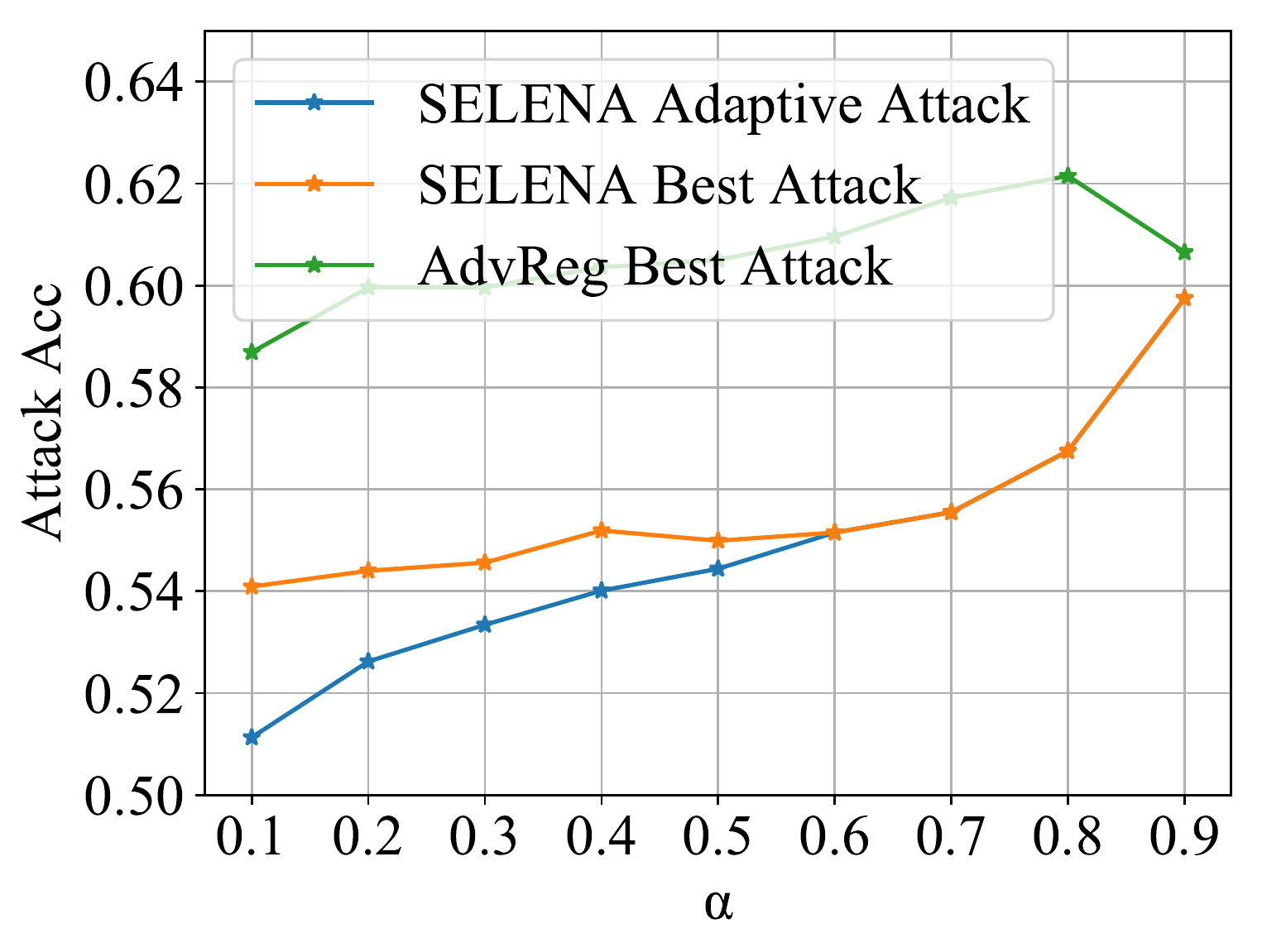}
        \end{minipage}
        }%
        \subfigure[CIFAR100]{
        \begin{minipage}[t]{0.3\linewidth}
        \centering
        \includegraphics[width=2.1in]{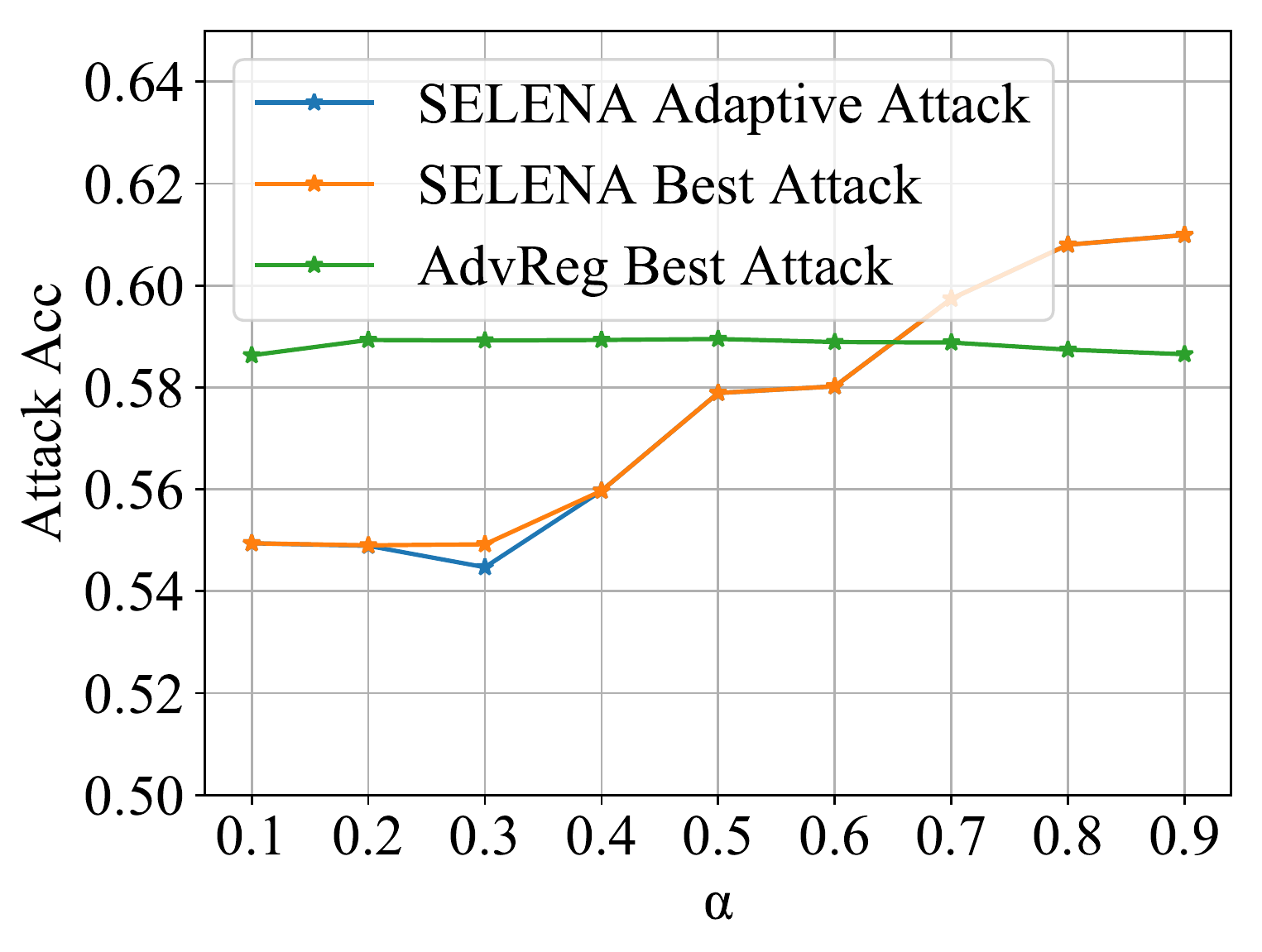}
        \end{minipage}
        }%
        \centering
    \caption{Impact of training set size used by attacker to train shadow \archname on adaptive attacks. $\alpha$ in Figure~\ref{fig:adaptive_attack} is the ratio of training samples used by attacker to train shadow \archname.}
    \label{fig:adaptive_attack}
    \end{figure*}

We evaluate this alternative soft labels approach on Purchase100 and present the classification accuracy as well as the best MIA attack accuracy (across multiple attack types) as a function of training epoch in Figure~\ref{fig:alphaparamter}. The best attack accuracy is the highest accuracy among direct single-query attack, label-only attacks and adaptive attacks~(attacker will estimate soft labels according to Equation~(\ref{eqa:alphaparameter})). We can now better control the trade-off between utility and membership privacy: if the desired MIA accuracy is no more than 56\%~(which is still lower than adversarial regularization, recall that the best MIA accuracy against adversarial regularization on Purchase100 is 57.4\%), then $\lambda = 0, 0.1, 0.2$ all satisfy the requirement. We observe that $\lambda=0.1$ increases classification accuracy by about 1\% and $\lambda=0.2$ increases classification accuracy by about 2\%  compared to \sysname.

\section{Detailed Analysis of Adaptive Attacks}
\label{appen:adaptiveattack}

We note that the number of training samples known by attacker will impact the quality of outputs generated by attacker's shadow \archname. For example, if the attacker knows all member samples, though not practical for the membership inference attack problem,  the confidence of shadow \archname's output is similar to that of defender's \archname's output. In contrast, knowing only part of member samples will lead to the result that the confidence of shadow \archname's output is lower than that of defender \archname's output.

To understand adaptive attacks thoroughly, we vary the number of member samples known by the attacker. Attacker will use all these member samples to train shadow \archname and the attacker goal is still to identify remaining unknown member samples and the baseline random guess is $50\%$ under the setting that the number of members and non-members used to train and evaluate the
attack model are the same.

Figure~\ref{fig:adaptive_attack} presents the performance of adaptive attacks as a function of size of training samples used by attacker to train shadow \archname including the best attack accuracy across multiple attacks~(direct single-query/label-only/adaptive attack for \sysname and direct single-query/label-only attack for adversarial regularization) as well as the adaptive attack for \sysname. $\alpha$ in Figure~\ref{fig:adaptive_attack} is the ratio of training samples used by attacker to train shadow \archname. Figure~\ref{fig:adaptive_attack} shows that for all three datasets, the adaptive attack accuracy increases as the increasing number of training samples used to train attacker's shadow \archname. Specifically, for Texas100 dataset, when $\alpha \leq 0.5$, the adaptive attack is lower than other two MIAs. Our \sysname performs well across different $\alpha$ settings: for Purchase100 and Texas100, the adaptive attack accuracy is lower than best MIA attack against adversarial regularization for all $\alpha$s. For CIFAR100, we can see that the adaptive attack accuracy against \sysname is lower than the best attack accuracy for CIFAR100 against adversarial regularization for $\alpha\leq 0.6$ and slightly higher (around 2\% at $\alpha=0.9$) than adversarial regularization for $\alpha\geq 0.7$.

\section{Comparison with Model Stacking}
\label{appendix:modelstacking}

    \begin{table}[ht]
    \caption{Comparison of \sysname and Model Stacking~(MS) against direct single-query attack.}
    \label{tab:modelstacking}
    \begin{tabular}{ccccc}
        \toprule
        dataset&defense & \tabincell{c}{acc on \\training \\set} & \tabincell{c}{acc on \\test set}& \tabincell{c}{best \\attack}\\
        \midrule
        \multirow{3}{*}{Purchase100} &\sysname &82.7\% & 79.3\% &53.3\%\\
        &\tabincell{c}{MS} &84.3\% &73.6\% &62.4\%\\
        \midrule
        \multirow{3}{*}{Texas100} &\sysname &58.8\% &52.6\% &54.8\%\\
        &\tabincell{c}{MS}&60.5\% &46.3\% &63.8\%\\
        \midrule
        \multirow{3}{*}{CIFAR100} &\sysname &78.1\% &74.6\% &55.1\%\\
        &\tabincell{c}{MS} &80.6\% &66.5\% &63.2\%\\
        \bottomrule
    \end{tabular}
    \end{table}

Table~\ref{tab:modelstacking} presents the comparison between our defense and model stacking against direct single-query attack. The classification accuracy of model stacking is lower than \sysname: 5.7\% lower on Purchase100, 6.3\% lower on Texas100 and 8.1\% lower on CIFAR100, while the direct single-query MIA accuracy against model stacking is higher than that against our defense: 
9.1\% higher on Purchase100, 9.0\% higher on Texas100 and 8.1\% higher on CIFAR100. This experiment supports our statements in Section~\ref{subsec:modelstacking}:~(1) Model stacking suffers a drop in test accuracy as it requires a disjoint subset of data for each module.~(2) Disjoint dataset partition alone is not enough to protect membership privacy if the outputs in the first layer are directly combined as inputs to the second layer. 

\end{document}